\documentclass[11pt]{article}
\usepackage[margin=1in]{geometry}
\usepackage[numbers]{natbib}

\usepackage{amsmath,amssymb,amsthm}
\usepackage[ruled,vlined,linesnumbered]{algorithm2e}
\usepackage{booktabs}
\usepackage{makecell}
\usepackage{hyperref}
\usepackage[table]{xcolor}
\definecolor{OurColor}{RGB}{0,50,110}
\usepackage{tikz}
\usepackage{pgfplots}
\pgfplotsset{compat=1.18}

\newtheorem{theorem}{Theorem}[section]
\newtheorem{lemma}[theorem]{Lemma}

\newtheorem{corollary}[theorem]{Corollary}
\newtheorem{observation}[theorem]{Observation}

\theoremstyle{definition}
\newtheorem{definition}[theorem]{Definition}
\newtheorem{problem}[theorem]{Problem}
\newtheorem{openquestion}[theorem]{Open Question}
\theoremstyle{remark}
\newtheorem{remark}[theorem]{Remark}

\newtheorem{fact}[theorem]{Fact}
\newtheorem*{theorem*}{Theorem}
\newtheorem*{corollary*}{Corollary}
\newtheorem*{proposition*}{Proposition}

\newcommand{\OPT}{\mathrm{OPT}}
\newcommand{\N}{\mathcal{N}}
\newcommand{\dotO}{\dot{O}}
\newcommand{\hatO}{\hat{O}}

\title{Submodular Ground-Set Pruning: Monotone Tightness and a Non-Monotone Separation}
\author{%
  Alan Kuhnle \\
  Department of Computer Science \& Engineering \\
  Texas A\&M University \\
  \texttt{kuhnle@tamu.edu}
}

\begin{document}
\maketitle

\begin{abstract}
Large-scale subset selection asks for a small useful set of examples, features,
sensors, seed users, or context passages from an enormous ground set.
Submodular maximization is a canonical model for such diminishing-returns
problems, but rapidly growing datasets make even linear-time algorithms ever costlier.
We study \emph{containment pruning}: first reduce the ground set to a smaller
core $P$, then require that $P$ contain a near-optimal feasible solution for
every downstream budget up to~$k$.
Prior work has formulated many heuristics, but
the theoretical limits of this preprocessing problem are largely unknown.
For monotone submodular objectives, we prove that $1-1/e$ is tight: greedy
achieves this containment factor, and no algorithm can beat it even with a
larger pruning budget.
For non-monotone objectives, we give the first
$1/2-\varepsilon$ containment algorithms under cardinality constraints and
extend the approach to knapsack constraints.
This $1/2$ factor exceeds the best known algorithmic ratio and the known
hardness threshold for non-monotone maximization, showing that pruning
can be provably easier than optimization.
Empirically, pruning lets an exact IP solver run on the reduced MaxCut
instance with a ${\approx}620\times$ speedup,
and
proof-of-concept experiments on LLM context selection demonstrate the utility of
non-monotone submodular proxies and our proposed containment algorithms.
\end{abstract}

\section{Introduction}
\label{sec:intro}

\paragraph{Large-scale subset selection.}
Many machine-learning and data-analysis tasks ask for a small useful subset of
a much larger ground set: seed users in a social network, sensors or experiment
locations, predictive features, representative documents, or context passages
for an LLM~\cite{KempeKleinbergTardos2003,KrauseSinghGuestrin2008,DasKempe2011,LinBilmes2011,KumariWangDasBilmes2024,KumariWangZhouBilmes2024}.
Submodular maximization is a canonical model for these problems because it
captures diminishing returns: once a subset already covers a concept, feature,
region, or piece of evidence, adding another similar item helps less.  The
ground sets in these applications keep growing: the NCBI Sequence Read Archive
has sustained exponential growth for over a decade~\cite{kodama2012sequence,katz2022sequence},
and machine-learning training corpora have grown at roughly $3.7\times$ per
year~\cite{villalobos2022data}.  Even theoretically efficient submodular
algorithms cannot avoid this growth: in the value-oracle model they require at
least linear time in the ground-set size~\cite{Kuhnle2021}.  Larger ground sets
therefore mean more oracle queries, higher latency, and larger downstream solver
instances.

\paragraph{Why pruning can help.}
When the target solution size $k$ is small compared with the ground-set size
$n$, much of the data is often redundant or irrelevant.  Documents may be near
duplicates, features may become correlated after a few predictors are selected,
retrieved passages may be distractors or paraphrases, and many graph vertices
may have interchangeable marginal value.  This suggests a two-stage workflow:
first reduce the large ground set to a smaller core $P \subseteq N$, then run a
more expensive optimizer on $P$ alone.  This idea appears in pruned
submodularity graphs, multi-stage selection, candidate reduction, learned
graph pruning, and many heuristics~\cite{ZhouOuyangBilmes2017,WeiBilmes2014,BalkanskiMirzasoleimanSinger2016,MirzasoleimanBadanidiyuruKarbasi2015,Manchanda2020,IrelandMontana2022,TianEtAl2024,NathKuhnle2025}.
The key question is not whether pruning can speed up a particular implementation,
but whether the pruned universe still contains high-quality feasible solutions
for a downstream optimizer to find.  Moreover, the allowable core size may be
set by memory, latency, parallelism, or the largest instance an exact solver can
handle.  Following Nath and Kuhnle~\cite{NathKuhnle2025}, we study
containment as a function of a pruning budget $p \ge k$, and ask for one core
that supports every downstream target size $k' \le k$.

\begin{problem}[Budget-Range Submodular Pruning]
\label{prob:pruning}
Given $N$, $f\colon 2^N \to \mathbb{R}_{\ge 0}$ (submodular), and
$k \in \mathbb{N}$, for a pruning budget $p \ge k$, find $P \subseteq N$
with $|P| \le p$ such that
\[
  \max_{T \subseteq P,\, |T| \le k'} f(T) \;\ge\; \alpha \cdot \OPT_{k'},
  \qquad \forall\, k' \le k,
\]
where $\OPT_{k'} = \max_{|T| \le k'} f(T)$.
We call $\alpha$ the \emph{containment factor}.
\end{problem}

\noindent The pruning algorithm does not need to identify the final solution.
It only has to output a smaller universe in which good solutions still exist.
This distinction is what makes containment different from ordinary
optimization.

\paragraph{The gap.}
Existing pruning methods have largely been heuristic:
pruned submodularity graphs~\cite{ZhouOuyangBilmes2017} and learned policies
such as GCOMB, LeNSE, and COMBHelper~\cite{Manchanda2020,IrelandMontana2022,TianEtAl2024}
reduce search spaces empirically, but do not give worst-case guarantees that a
good downstream solution remains in the pruned set.  The main theoretical
exception is Nath and Kuhnle's \textsc{QuickPrune}~\cite{NathKuhnle2025}, which
achieves $\alpha \ge 1/24-\varepsilon$ for monotone weakly submodular objectives
under a knapsack constraint.

Intuitively, pruning may be easier than direct optimization, as we are not
required to find an optimal solution, merely to contain it; but
the complexity of containment pruning is still
unclear even for monotone submodular functions, and essentially open for
non-monotone functions.
We investigate this complexity directly: can larger
pruning budgets beat the classical $1-1/e$ monotone barrier, and can
non-monotone objectives admit constant-factor budget-range containment?

\subsection{Contributions}

\begin{itemize}
\item \textbf{Tight monotone pruning.}
  We fully characterize monotone cardinality pruning: standard
  greedy gives simultaneous $(1-1/e)$-containment for every
  $k' \le k$~\cite{NemhauserWolsey1978,NathKuhnle2025}, and no algorithm can
  guarantee a larger factor for any pruning budget $p\ge k$, even when it may
  output any set $P$ with $|P|\le p$ (Theorem~\ref{thm:monotone-hardness}).
  The hardness extends to monotone $\gamma$-weakly DR functions
  and yields the tight $1-\kappa/e$ curvature barrier
  (Corollary~\ref{cor:curvature-hardness}).  We also give a near-linear
  threshold-greedy variant that obtains
  $(1-1/e-\varepsilon)$-containment with $O(k/\varepsilon)$ elements
  and $\tilde O(n/\varepsilon^2)$ value queries
  (Theorem~\ref{thm:fast-monotone}).

\item \textbf{Non-monotone pruning at $1/2$.}
  We give the first constant-factor budget-range containment guarantees
  for non-negative non-monotone submodular functions under cardinality
  constraints.  Sequential disjoint greedy and a one-pass window method both
  approach $\alpha=1/2$ while preserving every budget $k' \le k$
  (Theorems~\ref{thm:half-containment} and~\ref{thm:window-containment}).
  This factor strictly exceeds the best known non-monotone maximization ratio
  ($0.401$~\cite{BuchbinderFeldman2024}) and the best possible optimization
  ratio under standard assumptions ($0.478$~\cite{Qi2022}).  This is possible
  because containment only asks the pruned set to \emph{contain} a good feasible
  solution, not to extract it. Thus, containment pruning is provably
  easier than optimization in this setting.
 
\item \textbf{Extensions beyond cardinality.}
  The non-monotone analysis extends to knapsack constraints via sequential
  disjoint density-greedy (Theorem~\ref{thm:knapsack-half-containment}).
  The split-item analysis gives $(1/2-\varepsilon)\OPT_{B'}$ minus an
  explicit boundary-fragment term for every budget $B'\le B$.
  This term accounts for the last density-greedy item crossing the budget
  boundary.  Under a standard small-item condition, it is absorbed and the clean
  $1/2-\varepsilon$
  guarantee is recovered (Corollary~\ref{cor:knapsack-small-items}).
  The output has total cost $O(B/\varepsilon)$ and uses
  $O(n^2/\varepsilon)$ value queries.

\item \textbf{Empirical validation.}
  On MaxCut, pruning a $G(300,4500)$ instance to $p=50$ vertices lets an exact
  IP solver finish in 0.19s instead of 120s, a ${\approx}620\times$ speedup while
  preserving $\alpha=1.000$.  An adapted COMBHelper baseline is competitive on
  in-distribution random graphs but degrades on out-of-distribution SNAP
  graphs; our provable methods require no training data.  For LLM context
  selection, $38.7\%$ of HotpotQA questions exhibit non-monotone context
  effects.  A non-monotone proxy
  $\hat{f}(S)=\mathrm{FL}(S)-\theta(|S|)$ shows promise as a context-selection
  objective: on a 50-question MuSiQue pilot, relevance-gated pruning improves
  gold-passage recall by $+9.2$ percentage points at $k=10$.
\end{itemize}

\begin{table}[t]
\centering
\small
\setlength{\tabcolsep}{2.5pt}
\renewcommand{\arraystretch}{1.12}
\begin{tabular}{@{} >{\raggedright\arraybackslash}p{2.55cm}
  >{\raggedright\arraybackslash}p{2.05cm}
  >{\raggedright\arraybackslash}p{3.05cm}
  >{\raggedright\arraybackslash}p{4.55cm} @{}}
\toprule
Question & Setting & Prior & This paper \\
\midrule
Does extra pruning budget help?
  & Monotone card.
  & Greedy: $1-1/e$ at $p=k$
  & {\color{OurColor}$1-1/e$ for every $p\ge k$, with matching hardness}
  \\[3pt]
Can pruning be faster?
  & Monotone card.
  & Greedy: $O(nk)$ queries
  & {\color{OurColor}$1-1/e-\varepsilon$ with $|P|=O(k/\varepsilon)$ and
     $\tilde O(n/\varepsilon^2)$ queries}
  \\[3pt]
Can we prune non-monotone objectives?
  & Non-monotone card.
  & No constant-factor budget-range guarantee
  & {\color{OurColor}$1/2-\varepsilon$ via sequential disjoint greedy or
     window pruning}
  \\[3pt]
Can this extend beyond cardinality?
  & Non-monotone knapsack
  & No non-monotone guarantee; NK25 gives monotone weak-submodular $1/24$
  & {\color{OurColor}$1/2-\varepsilon$ minus boundary term; clean under small items}
  \\
\bottomrule
\end{tabular}
\caption{Main theoretical results, organized by the questions they answer.
{\color{OurColor}Blue} marks new contributions.
The first row closes monotone cardinality pruning: extra pruning budget cannot
beat greedy's $1-1/e$ factor.  For non-monotone cardinality, the gap between
$1/2$ and the inherited $1-1/e$ upper bound remains open; the knapsack result
has the boundary-term caveat shown in the table, with a clean small-item
corollary.}
\label{tab:summary}
\end{table}

\paragraph{Paper organization.}
Section~\ref{sec:prelim} collects notation and background.
Section~\ref{sec:monotone} presents the monotone hardness result.
Section~\ref{sec:nonmonotone} develops both non-monotone containment
algorithms.
Section~\ref{sec:experiments} reports experiments on LLM context selection,
MaxCut, and synthetic non-monotone instances.
Section~\ref{sec:conclusion} discusses open problems.
Appendix~\ref{sec:related} compares containment with bicriteria approximation,
composable coresets, streaming summaries, and two-stage optimization.

\section{Preliminaries}
\label{sec:prelim}

This section fixes notation and collects the density-dependent
maximization bounds that serve as reference points throughout the paper.
The two key quantities to keep in mind are the \emph{containment factor}
$\alpha$ (how well a pruned set represents the full ground set) and the
\emph{density} $c = k/n$ (which governs both the best
achievable approximation and the composition bound of \S\ref{sec:monotone}).
Readers familiar with submodular maximization may skim this section and
refer back as needed.

Let $N$ be a finite ground set with $|N| = n$.
A function $f\colon 2^N \to \mathbb{R}$ is \emph{submodular} if for all
$A \subseteq B \subseteq N$ and $x \notin B$,
$f(x \mid A) \ge f(x \mid B)$,
where $f(x \mid A) = f(A \cup \{x\}) - f(A)$ denotes the marginal gain.
We assume $f$ is non-negative; note that $f(\emptyset) = 0$ is \emph{not} required---our proofs carry $f(\emptyset)$ explicitly where it appears.
The function is \emph{monotone} if $f(A) \le f(B)$ whenever $A \subseteq B$.
For $\gamma \in (0, 1]$, the function is \emph{$\gamma$-weakly DR}
(diminishing returns) if
$f(x \mid A) \ge \gamma \cdot f(x \mid B)$
for all $A \subseteq B$ and $x \notin B$;
at $\gamma = 1$ this is ordinary submodularity.
The \emph{total curvature} of a monotone function is
$\kappa = 1 - \min_{j \in N} f(j \mid N \setminus \{j\})/f(\{j\})$;
at $\kappa = 0$ the function is modular.

We write $O$ for an optimal solution: $O \in \arg\max_{|T| \le k} f(T)$,
and $\OPT = f(O)$.
The \emph{density} of the constraint is $c = k/n$.

\subsection{Known density-dependent results}

\paragraph{Monotone.}
The Measured Continuous Greedy (MCG)---a continuous-relaxation algorithm
that iteratively pushes a fractional solution toward the multilinear
extension's gradient---achieves approximation ratio
$\alpha^*(c) = 1 - (1-c)^{1/c}$ at density $c$, which is tight
(Feldman \& Kuhnle~\cite{FeldmanKuhnle2025}, Theorems~4--5).
Concretely: $\alpha^*(c) \to 1-1/e \approx 0.632$ as $c \to 0$,
$\alpha^*(1/2) \approx 0.750$, $\alpha^*(0.9) \approx 0.935$,
and $\alpha^*(c) \to 1$ as $c \to 1$.
This function governs the composition bound in \S\ref{sec:monotone}
and tells us how much ``easier'' maximization becomes as the pruned set
shrinks the effective ground set.

\paragraph{Non-monotone.}
Buchbinder, Feldman, Naor, and Schwartz (BFNS)~\cite{BFNS2014} give the
density-dependent bound
\[
  \alpha^*_{\mathrm{nm}}(c) = \max\!\left\{
    \frac{1}{e} + 0.004,\;
    \left(1 + \frac{1}{2\sqrt{c(1-c)}}\right)^{\!-1}
  \right\} - o(1).
\]
Concretely, this function is nearly flat across the density regime
relevant to pruning:
$\alpha^*_{\mathrm{nm}}(0.01) \approx 0.372$,
$\alpha^*_{\mathrm{nm}}(0.1) \approx 0.400$,
$\alpha^*_{\mathrm{nm}}(0.2) \approx 0.444$,
rising to $1/2$ at $c = 1/2$ (tight by the symmetry
gap~\cite{GharanVondrak2011,Vondrak2013}).
Buchbinder and Feldman~\cite{BuchbinderFeldman2024} subsequently achieved
$0.401$ for all down-closed constraints (density-independent).
The combined best is $\max(0.401, \alpha^*_{\mathrm{nm}}(c))$.
The strongest hardness is $0.478$ due to Qi~\cite{Qi2022},
improving the earlier $0.491$ of Gharan and Vondr\'{a}k~\cite{GharanVondrak2011}.
The small spread $[0.401, 0.5]$ across all densities is why the composition
bound does essentially no work in the non-monotone setting
(\S\ref{sec:nonmonotone}).

One additional non-monotone baseline---the three-term CDG guarantee of~\cite{BFNS2014}---is used only in the warm-start analysis and is stated in Appendix~\ref{app:deferred-prelim}.

\section{Monotone Pruning}
\label{sec:monotone}

\paragraph{Does extra budget help?}
The simplest way to analyze pruning is \emph{composition}: prune to $p$ elements, then run the best optimization algorithm on the pruned set.
If the pruned set has containment factor $\alpha$ and the optimizer has approximation ratio $\alpha^*(k/p)$, the combined guarantee is at most $\alpha \cdot \alpha^*(k/p)$.
Since this cannot exceed the information-theoretic limit on the original problem, we obtain $\alpha \le \alpha^*(k/n)/\alpha^*(k/p)$.
At $p=k$ this recovers $\alpha \le 1-1/e$, but it weakens as $p$ grows.
The following theorem resolves the remaining question directly: extra budget truly cannot help.

\begin{theorem}[Pruning Hardness for Monotone $\gamma$-Weakly DR Functions]
\label{thm:monotone-hardness}
For every $\gamma \in (0, 1]$, every $\varepsilon > 0$,
and every polynomial $C(\cdot)$,
there exists $k_0 = k_0(\gamma, \varepsilon, C)$ such that for all $k \ge k_0$
and \emph{every} $p \ge k$:
no (possibly randomized) algorithm making at most $C(|N|)$ value queries
and outputting $P \subseteq N$ with $|P| \le p$ can guarantee
\[
  \mathbb{E}\!\left[\max_{S \subseteq P,\, |S| \le k} f(S)\right]
  > (1 - e^{-\gamma} + \varepsilon) \cdot \OPT_k
\]
for all monotone $\gamma$-weakly DR $f\colon 2^N \to \mathbb{R}_{\ge 0}$
with $|N| = \lceil 60p/\varepsilon \rceil$.
At $\gamma = 1$ (full submodularity), this recovers the tight bound
$\alpha \le 1 - 1/e$.
In particular, extra budget beyond $k$ provably does not improve the containment
factor for any constant or growing multiple of~$k$.
\end{theorem}

\begin{proof}[Proof sketch (full proof in Appendix~\ref{app:monotone-proof})]
The core idea is a cancellation that neutralizes extra budget.
We set $|N| = \lceil 3p/\varepsilon' \rceil$ and plant a random optimum
$\tilde{T}$ of size~$k$ using the construction of Harshaw et
al.~\cite{HarshawFeldmanWardKarbasi2019}.
Any algorithm's output $P$ (with $|P|\le p$) is indistinguishable from the
null instance, so $P$ is effectively fixed before $\tilde{T}$ is drawn.
The expected overlap is then at most
$\mathbb{E}[|P \cap \tilde{T}|]\le pk/|N|\le \varepsilon'k/3$,
which is \emph{independent of $p$}: the pruning-budget factor in the
numerator cancels with the one in the denominator ($|N| \propto p$).
A concentration bound shows that with high probability, $P$ has too
little overlap with $\tilde{T}$ for any $k$-subset of $P$ to beat
the null-instance value $1 - e^{-\gamma} + O(\varepsilon)$.
\end{proof}

\begin{remark}
The proof uses two separate hypergeometric tails.  First, along the
null-instance transcript, every query is hidden unless it has unusually large
overlap with the random planted set; a two-regime bound covers both
$p=O(k)$ and $p\gg k$.  Second, after the null transcript fixes the output
$P$, the expected overlap
$\mathbb{E}[|P\cap\tilde T|]\le \varepsilon' k/3$ is independent of~$p$.
Thus the extra budget cancels in the hiding argument rather than revealing
more of the planted optimum.
\end{remark}

Real-world submodular functions---coverage, influence, facility
location---often have total curvature $\kappa < 1$, meaning elements
retain marginal value even in the presence of many others.
For maximization, low curvature yields strictly better approximation
($1 - \kappa/e$ vs.\ $1-1/e$).
One might hope that low curvature similarly helps pruning.
The next result shows it does not: greedy at $p = k$ is already optimal
for every curvature level.

\begin{corollary}[Pruning Hardness with Bounded Total Curvature]
\label{cor:curvature-hardness}
Let $\kappa \in (0, 1]$ denote total curvature:
$\kappa = 1 - \min_{j \in N} f(j \mid N \setminus \{j\})/f(j)$.
For every $p \ge k$ and every $\varepsilon > 0$,
no algorithm using polynomially many value queries can prune to
$p$ elements with containment
$\alpha > 1 - \kappa/e + \varepsilon$
for all monotone submodular functions with total curvature at most~$\kappa$.
Since the algorithm of Sviridenko, Vondr\'{a}k \& Ward~\cite{SviridenkoVondrakWard2017} achieves $\alpha = 1 - \kappa/e$,
this bound is tight.
\textup{(}At $\kappa = 0$ the function is modular, so exact containment is trivial.\textup{)}
\end{corollary}

\begin{proof} See Appendix~\ref{app:curvature-proof}. \end{proof}

\noindent For general monotone submodular functions, greedy achieves $1-1/e$ and the hardness matches it at every pruning budget $p \ge k$.
Table~\ref{tab:monotone} (Appendix~\ref{app:curvature-proof}) records how this compares with the composition benchmark.

\begin{remark}[Fast budget-range pruning]
\label{sec:fast-monotone}
Standard greedy attains $\alpha = 1-1/e$ in $O(nk)$ queries via its prefix property~\cite{NathKuhnle2025}.
A small-budget grid plus a geometric sweep of $O((1+\log k)/\varepsilon)$ calls to the threshold greedy of~\cite{BadanidiyuruVondrak2014} achieves exact-size $(1-1/e-\varepsilon)$-containment for all $k' \le k$ simultaneously, with $|P| = O(k/\varepsilon)$ and $\tilde O(n/\varepsilon^2)$ total queries (Theorem~\ref{thm:fast-monotone}, Appendix~\ref{app:fast-monotone}).
The proof first obtains a slightly oversized witness and then thins it to size $k'$; monotone submodularity loses only an $O(\varepsilon)$ multiplicative factor.
\end{remark}

\section{Non-Monotone Pruning}
\label{sec:nonmonotone}

\paragraph{Why composition fails for non-monotone functions.}
In the monotone setting, composition at least gives the right answer at $p=k$
(even though it fails at larger budgets).
For non-monotone objectives, composition is nearly useless even at $p=k$.
The reason is that non-monotone approximation ratios barely change with density:
the best ratio ranges from $0.401$ at low density to $0.5$ at high density.
Since the numerator and denominator of the composition bound
$\alpha \le \alpha^*_{\mathrm{nm}}(k/n) / \alpha^*_{\mathrm{nm}}(k/p)$
are nearly equal, the bound stays close to~$1$ for moderate pruning
budgets and approaches the vacuous bound~$1$ as $p$ grows.
Entirely new techniques are needed.

\paragraph{Can we prune non-monotone functions?}
We show that $O(k)$ elements suffice for $1/2$-factor containment.
For the top budget $k$, storing the output of the $0.401$-approximation
of~\cite{BuchbinderFeldman2024} gives only a single-budget optimization
baseline; our guarantee is budget-range containment.

\begin{theorem}[Universal $1/2$ Containment]
\label{thm:half-containment}
Let $f$ be a non-negative submodular function on ground set $N$, and let
$k \le n$.
For any $\varepsilon > 0$, there exists a set $G \subseteq N$ with
$|G| = O(k/\varepsilon)$ such that for every $k' \le k$,
(if $|N| < \lceil 1/\varepsilon \rceil k$, take $G = N$;
containment is exact and $|G| < \lceil 1/\varepsilon \rceil k = O(k/\varepsilon)$)
\[
  \max_{T \subseteq G,\, |T| \le k'} f(T)
  \;\ge\; \left(\tfrac{1}{2} - \varepsilon\right) \cdot \OPT_{k'},
\]
where $\OPT_{k'} = \max_{|S| \le k'} f(S)$.
Moreover, $G$ can be constructed in $O(nk/\varepsilon)$ oracle queries
($\ell = \lceil 1/\varepsilon \rceil$ greedy passes, each scanning $O(n)$
elements at $k$ steps).
\end{theorem}

\begin{proof}[Proof sketch (full proof in Appendix~\ref{app:half-containment-proof})]
\emph{Construction.}
Run $\ell$ independent greedy runs $G_1, \ldots, G_\ell$, each of size $k$, on disjoint portions of the ground set.
The pruned set is $G = \bigcup_i G_i$ with $|G| = \ell k$.
Some optimal elements may lie outside $G$; call these $\hatO = O \setminus G$ (``missed'' elements, $b = |\hatO|$), and let $A_i$ be the first $b$ elements of each greedy run~$i$.

\emph{Averaging.}
Because the runs are disjoint, their collective ``damage'' to $f(O)$ is bounded:
the total loss $\sum_i [f(O) - f(O \cup A_i)]$ is at most $f(O)$.
By averaging, some run $i^*$ satisfies $f(O \cup A_{i^*}) \ge (1-1/\ell)f(O)$.

\emph{Greedy domination.}
The key insight is that greedy run $i^*$ had access to the missed elements $\hatO$ as candidates but chose $A_{i^*}$ instead.
So each greedy pick has marginal gain at least as large as the corresponding missed element.
Telescoping this chain gives $2f(A) - f(\emptyset) \ge f(\hatO \cup A)$: the greedy prefix captures at least half the value of the prefix augmented with the missed elements.

\emph{Combining.}
A submodularity decomposition yields $\max(f(A), f(\dotO \cup A)) \ge (1-1/\ell) f(O)/2$, where $\dotO = O \cap G$ is the part of OPT already in $G$.
Both candidates are feasible subsets of $G$ with size at most $k'$,
giving $\alpha \ge (1/2)(1-1/\ell)$.
The argument applies uniformly to every target budget $k' \le k$.
\end{proof}

\begin{remark}[Where the $1/2$ comes from, and why it is robust]
The $1/2$ factor arises from two tight ingredients:
(a) the greedy domination bound $2f(A) - f(\emptyset) \ge f(\hatO \cup A)$,
where the non-monotone marginal chain only gives a factor of~$2$
(vs.\ the $(1-1/e)^{-1}$ one gets for monotone);
and (b) the $\max \ge \mathrm{sum}/2$ split in the final step.
Neither bound depends on $|G|$, so the $1/2$ is robust to increasing the
pruning budget---with pruning budget $p = \ell k$, the
containment factor is $\frac{1}{2}(1-1/\ell)$: $\ell = 2$ gives
$\alpha \ge 1/4$, while $\ell \to \infty$ recovers $1/2$.
Note this is \emph{not} bicriteria: both candidate solutions $A$ and
$\dotO \cup A$ satisfy $|T| \le k$, improving over the bicriteria result
of Feldman and Kuhnle~\cite{FeldmanKuhnle2025} ($(1/2-\varepsilon)$ with
output size $\sim 1.44k$).
The $1/2$ is reminiscent of the unconstrained non-monotone hardness
of~\cite{FeigeMirrokniVondrak2011}, but the analogy is suggestive rather
than exact, and closing the gap to $1-1/e$ remains open
(\S\ref{sec:conclusion}).
\end{remark}

A natural warm-start approach via the BFNS continuous double greedy~\cite{BFNS2014} peaks at $4/9 \approx 0.444$ and cannot beat $1/2$; see Appendix~\ref{app:warm-start} for details.

\paragraph{A second proof via window containment.}
\label{sec:window-greedy}
Sequential disjoint greedy requires $\ell$ independent passes over the ground set.
We now give a structurally distinct proof of $1/2$-containment using a single-pass
\emph{window} algorithm (Algorithm~\ref{alg:window-greedy}, Appendix~\ref{app:window-analysis})
that achieves the same factor with a larger accumulated-window output.
At each step, we maintain a window $W_t$ of the top-$\omega k$ elements by marginal gain
and pick one uniformly at random; the window parameter $\omega \ge 1$
controls the width of exploration (larger $\omega$ means a wider pool of
candidates per step, at the cost of retaining all windows and hence a
larger pruned set of size at most $k+\omega k^2$).
This simultaneously guarantees \emph{domination} (every window element has
marginal gain at least as large as any non-window element, so a random
pick dominates any missed OPT element) and \emph{degradation control}
(the OPT value is not destroyed by the elements we commit to:
$\mathbb{E}[f(O \cup S_m)] \ge e^{-1/\omega} f(O)$).

\begin{theorem}[Window Containment --- Unconditional]
\label{thm:window-containment}
For any non-negative submodular $f$ and window parameter $\omega \ge 1$,
Algorithm~\ref{alg:window-greedy} returns an accumulated-window set
$P_k \subseteq N$ with $|P_k| \le k+\omega k^2$ such that, for every
$k' \le k$,
\[
  \mathbb{E}\!\left[\max_{|T| \le k',\; T \subseteq P_k} f(T)\right]
  \;\ge\; \frac{e^{-1/\omega}}{2} \cdot \OPT_{k'}
  \qquad\text{(asymptotically as $k \to \infty$)}.
\]
The finite-$k$ lower bound from the proof is $(1 - 1/(\omega k))^{k}/2
  = e^{-1/\omega}\bigl(1 - O(1/(\omega^2 k))\bigr)/2$;
in particular, $\alpha \to 1/2$ as $\omega \to \infty$.
Total query complexity: $O(nk)$.
Full proof in Appendix~\ref{app:window-analysis}; see Table~\ref{tab:window} there for the
guarantee as a function of $\omega$ (at $\omega = 5$, $\alpha \ge 0.410 > 0.401$).
\end{theorem}

\begin{remark}[Window Pruning]
\label{rem:window-pruning}
Window containment is structurally distinct from sequential disjoint greedy
(Theorem~\ref{thm:half-containment}): it makes a single wide-exploration pass
rather than $\ell$ independent narrow searches.
Both converge to the same $1/2$ barrier, suggesting it may be a fundamental limit.
\end{remark}

\paragraph{Knapsack extension.}
\label{sec:knapsack-extension}
The SDG argument extends beyond cardinality constraints to knapsack constraints, where each element $u$ has a positive cost $p_u$ and a feasible solution must have total cost at most~$B$.
(We assume $p_u \le B$ for every element, since items too expensive to use alone can be discarded.)
The key change is replacing greedy---which picks the element with the highest marginal gain---by \emph{density-greedy}, which picks the element with the highest \emph{value-per-unit-cost}: at each step, select $u$ maximizing $f(u \mid A)/p_u$, and keep it if adding it stays within cost $3B$.
Running $\ell$ disjoint density-greedy passes (with dummy zero-value items to fill each run to cost $2B$, following~\cite{FeldmanKuhnle2025}) produces the pruned set.

\begin{theorem}[Knapsack budget-range containment]
\label{thm:knapsack-half-containment}
For any $\varepsilon \in (0, 1/2)$, sequential disjoint density-greedy with
$\ell = \lceil 1/\varepsilon \rceil$ runs produces a set $P \subseteq N$
with $c(P) \le 3 \ell B = O(B/\varepsilon)$ such that for every
$B' \in (0, B]$ there exists $Q \subseteq P$ with
\[
  c(Q) \;\le\; B', \qquad
  f(Q) \;\ge\; \left(\tfrac{1}{2} - \varepsilon\right) \OPT_{B'}
  \;-\; \Delta_{B'},
\]
where $\Delta_{B'} \ge 0$ is a boundary-fragment loss that arises when the
density-greedy prefix crosses the budget threshold mid-item.  When no item
is split, $\Delta_{B'} = 0$.
Total query complexity: $O(n^2 \ell) = O(n^2/\varepsilon)$
($\ell$ adaptive density-greedy runs, each scanning $O(n)$ candidates per step).
\end{theorem}

The boundary term $\Delta_{B'}$ arises when the density-greedy prefix ``overshoots'' the budget---the last item added may cost more than the remaining budget allows.
In that case, only a fraction $\gamma$ of the item's value is captured, and the remaining value is lost.
Formally, $\Delta_{B'} = \gamma_{i^\star} f(\{a_{j^\star}^{(i^\star)}\})$,
where $a_{j^\star}^{(i^\star)}$ is the item that crosses the budget boundary;
see Appendix~\ref{app:knapsack-details} for precise definitions.
The split item may cost more than~$B'$ itself, so $\Delta_{B'}$ cannot be bounded simply by the value of the best affordable item.

\begin{corollary}[Small-item knapsack containment]
\label{cor:knapsack-small-items}
Fix a budget range $\mathcal{B}\subseteq(0,B]$ and let
$B_{\min}=\inf \mathcal{B}>0$.
Run sequential disjoint density-greedy with a sufficiently small internal
accuracy parameter.  Suppose that the item sizes are small at the density
prefix scale used in the proof: for every $B'\in\mathcal{B}$, if the proof for
$B'$ selects a split boundary item $a_{j^\star}^{(i^\star)}$ with residual
outside-core cost $\beta=c(O\setminus P)$, then
\[
  p_{a_{j^\star}^{(i^\star)}} \le \eta \beta .
\]
In particular, this is the usual no-large-items regime when the relevant
residual prefixes have cost comparable to the query budget; over a fixed range
it is ensured by $p_{\max}\le \eta B_{\min}$ together with
$\beta=\Omega(B_{\min})$ for the budgets under consideration.
There is a universal constant $C$ such that, if $\eta\le \varepsilon/C$, then
for every $B'\in\mathcal{B}$ there exists $Q\subseteq P$ with
\[
  c(Q)\le B',
  \qquad
  f(Q)\ge \left(\tfrac12-\varepsilon\right)\OPT_{B'} .
\]
\end{corollary}

The proof adapts the sequential disjoint greedy argument of Theorem~\ref{thm:half-containment} to density-greedy with item splitting,
using the machinery of~\cite{FeldmanKuhnle2025}.
The proof sketch, simultaneous budget-range guarantee, knapsack hardness argument, and boundary-term discussion appear in Appendix~\ref{app:knapsack-details};
the full proof is in Appendix~\ref{app:knapsack-proof}.

\section{Experiments}
\label{sec:experiments}

In this section, we ask whether containment cores are useful preprocessing
objects.  We show three messages.  First, on MaxCut, our pruning algorithms
usually retain an optimal or near-optimal $k$-cut and can shrink an exact IP
solve by orders of magnitude.  Second, on HotpotQA, LLM context selection is
genuinely non-monotone: adding passages can hurt answer quality, and pruning a
non-monotone proxy improves downstream F1/EM while using fewer passages.
Third, on MuSiQue, the task has the containment structure pruning is designed
for: keep evidence-bearing passages in a reduced pool while filtering
distractors. We show that relevance-gated submodular diversity improves gold-passage recall
over relevance ranking alone.

Section~\ref{sec:exp-maxcut} studies MaxCut on random and real-world graphs,
including exact solving on pruned instances.  Section~\ref{sec:exp-llm} studies
LLM context selection.  Appendices~\ref{app:snap-results} and~\ref{app:separation}
report COMBHelper details and synthetic separations.
Unless stated otherwise, an algorithm receives a ground set~$\N$,
cardinality~$k$, and pruning budget $p = \omega k$, outputs
$P$ with $|P| \le p$, and then a $k$-subset is extracted from~$P$
by greedy or by exact IP.  When exact optimization is feasible, the containment
ratio is $\alpha = f(S^*_P)/\OPT_k$, where $S^*_P$ is the best $k$-set inside
the pruned universe.  When OPT is unavailable, we explicitly say which
full-ground-set reference is used.

\paragraph{Algorithms.}
We compare seven methods:
\textbf{SeqDisjoint} (Theorem~\ref{thm:half-containment}): $\omega$ disjoint greedy runs of size~$k$;
\textbf{WindowMax} / \textbf{WindowRand}: practical window heuristics with
argmax / random selection from top-$\omega k$ marginals;
\textbf{QuickPrune}~\cite{NathKuhnle2025}: streaming threshold pruning (monotone only---no guarantee for non-monotone);
\textbf{StdGreedy}: standard greedy for $p$ steps;
\textbf{Random}: uniform random $p$ elements;
\textbf{COMBHelper}~\cite{TianEtAl2024}: a learned GNN pruning baseline,
adapted to MaxCut; architecture, training, seed variance, and
out-of-distribution behavior are deferred to Appendix~\ref{app:snap-results}.

\subsection{MaxCut on random and real-world graphs}
\label{sec:exp-maxcut}

MaxCut asks whether pruning can preserve good solutions for a standard
non-monotone objective.  The objective is
$f(S) = |\{(u,v) \in E : u \in S, v \notin S\}|$.
It is non-monotone because adding a vertex can remove edges from the cut.
We test on Erd\H{o}s--R\'{e}nyi random graphs $G(n,m)$,
planted partition graphs, and SNAP social networks~(Facebook, $n=4039$, $m=88{,}234$).
For the small instances in Table~\ref{tab:maxcut-small}, OPT is computed
exactly via integer programming (CBC).  For the larger SNAP instances, OPT is
estimated by greedy on the full ground set, so reported $\alpha$ values are
upper bounds on the true containment ratio rather than certified guarantees.
The table also includes COMBHelper on its in-distribution Erd\H{o}s--R\'{e}nyi
setting; we defer its teacher/student split and transfer behavior to the
appendix.

\begin{table}[h]
\centering
\small
\begin{tabular}{@{}l ccc ccc@{}}
\toprule
& \multicolumn{3}{c}{Random $G(100, 1000)$, $k=15$}
& \multicolumn{3}{c}{Planted$(200, 20)$} \\
\cmidrule(lr){2-4} \cmidrule(lr){5-7}
Algorithm & $\omega=2$ & $\omega=3$ & $\omega=5$
          & $\omega=2$ & $\omega=3$ & $\omega=5$ \\
\midrule
SeqDisjoint & 0.997 & \textbf{1.000} & \textbf{1.000}
            & \textbf{1.000} & \textbf{1.000} & \textbf{1.000} \\
WindowMax   & \textbf{1.000} & 0.999 & \textbf{1.000}
            & \textbf{1.000} & \textbf{1.000} & \textbf{1.000} \\
WindowRand  & 0.975 & 0.967 & 0.984
            & 0.988 & 0.983 & 0.982 \\
QuickPrune  & 0.868 & 0.878 & 0.853
            & 0.865 & 0.876 & 0.874 \\
StdGreedy   & 0.997 & \textbf{1.000} & \textbf{1.000}
            & \textbf{1.000} & \textbf{1.000} & \textbf{1.000} \\
COMBHelper$^\dagger$ & 0.997 & \textbf{1.000} & \textbf{1.000}
            & -- & -- & -- \\
Random      & 0.877 & 0.921 & 0.966
            & 0.699 & 0.805 & 0.902 \\
\bottomrule
\end{tabular}
\caption{MaxCut containment on small instances, where exact IP gives the true
denominator $\OPT_k$.
SeqDisjoint (Thm~\ref{thm:half-containment}) and WindowMax
consistently achieve $\alpha \ge 0.99$, far above the worst-case $1/2$
guarantee.
Results are stable across graph topologies: over 5~ER seeds,
SeqDisjoint and WindowMax have mean $\alpha \ge 0.996$
(std $\le 0.004$), while the random baseline has std in
$[0.013,0.040]$.
COMBHelper$^\dagger$ is the adapted learned baseline on the ER setting; teacher
and student give the same values there, so we show one compact row and defer
the full learned-pruning comparison to Appendix~\ref{app:snap-results}.
QuickPrune, designed for monotone functions, plateaus at $\alpha \approx 0.87$
regardless of budget: once the non-monotone marginal sequence turns negative,
its streaming threshold accepts too few elements
($|P| \approx 18\text{--}27 \ll p$).}
\label{tab:maxcut-small}
\end{table}

\paragraph{SNAP social networks and exact solving.}
The same pattern persists on larger graphs, but the interpretation changes.
On Facebook ($n=4039$) and Wiki-Vote ($n=7115$) with $k=50$, exact OPT is
unavailable; using the full-ground-set greedy solution as the reference,
SeqDisjoint, WindowMax, and StdGreedy all obtain ratio $1.000$ at every tested
budget, while QuickPrune plateaus near $0.87$ (details in
Appendix~\ref{app:snap-results}).  These are empirical ratios against the greedy
reference, not certified OPT ratios.

Pruning also reduces solver size.  On a medium instance $G(300,4500)$ with
$k=25$, where full IP still finishes, solving MaxCut on the full graph takes
120\,s.  After SeqDisjoint pruning to $p=50$ vertices ($\omega=2$), the same IP
solve on the reduced instance takes 0.19\,s---a \textbf{${\approx}620\times$ speedup}
with $\alpha=1.000$.

\paragraph{Baselines and separations.}
COMBHelper is competitive in distribution, but it is a learned baseline that
requires exact training labels; Appendix~\ref{app:snap-results} gives the
teacher/student results, seed-variance check, and out-of-distribution SNAP
comparison.  The appendix also reports synthetic interference-coverage
instances where SeqDisjoint contains the true optimum more often than a single
greedy run (78\% vs.\ 60\% at $n=20$, $k=3$, $\omega=2$;
Appendix~\ref{app:separation}).

\subsection{Non-monotone LLM context selection}
\label{sec:exp-llm}

This subsection separates two roles for LLM context selection.  HotpotQA tests
whether context selection is genuinely non-monotone and whether a cheap
submodular proxy can guide pruning toward higher-quality reader inputs.
MuSiQue then tests the containment claim: when supporting evidence is dispersed
across hops, a pruner should keep evidence-bearing passages available for a
downstream selector or reader rather than output the final context itself.

Adding context passages to an LLM prompt initially improves answer quality
but degrades it beyond a threshold---the ``lost in the middle''
effect~\cite{LiuLinHewitt2024,DuEtAl2025}.
This makes context selection a natural source of non-monotone objectives:
more candidate passages can improve coverage, but more prompt context can
also dilute attention or introduce distracting evidence.

\paragraph{A non-monotone submodular proxy.}
We introduce a \emph{non-monotone submodular proxy}
$\hat{f}(S) = \mathrm{FL}(S) - \theta(|S|)$, where $\mathrm{FL}$ is
facility location on TF-IDF passage similarities and $\theta$ is a convex,
non-decreasing penalty fitted via isotonic regression from LLM evaluation
data.
Departing from prior submodular-LLM work
\cite{KumariWangDasBilmes2024, KumariWangZhouBilmes2024,
AgarwalKillamsetty2025, LinBilmes2011},
which uses monotone facility location exclusively, our proxy captures the
empirically observed quality degradation from context overload.
The function $\hat{f}$ is submodular (proof in Appendix~\ref{app:llm-details}).

\paragraph{HotpotQA proxy containment.}
We test on HotpotQA (distractor setting) with expanded ground sets
of $n = 50$ ($k = 3$) and $n = 500$ ($k = 5$);
setup details are in Appendix~\ref{app:llm-details}.
On the fitted proxy, SeqDisjoint and WindowMax achieve $\alpha = 1.0$ at every budget
(Table~\ref{tab:llm-proxy}, Appendix~\ref{app:snap-results});
a larger $n = 500$ FL-only stress test in Appendix~\ref{app:snap-results}
shows that $50\times$ compression preserves FL containment $\alpha = 1.000$.

\paragraph{Non-monotonicity prevalence and LLM validation.}
Among 200 HotpotQA questions evaluated with a local 120B-parameter LLM,
$77/199=38.7\%$ have some shorter prefix with higher F1 than the full
10-passage context, confirming that non-monotonicity is common enough that
monotone-only pruning can miss the relevant regime.
To verify that proxy containment translates to real answer quality,
we evaluate the final $k$-subsets with the LLM (Table~\ref{tab:llm-validation}).
At $\omega = 2$ ($p = 6$), proxy-guided greedy selects only $k = 3$~passages
yet achieves F1/EM $=0.599/0.42$, surpassing the all-original baseline
(F1/EM $=0.556/0.40$, all 10~passages).  This supports the proxy claim:
pruning can remove distractors while preserving enough evidence for the reader.
QuickPrune, which collapses on non-monotone inputs,
scores only F1 $= 0.262$---no better than random.

\paragraph{Relevance-gated passage selection (MuSiQue).}
HotpotQA supports the proxy claim, but its gold passages are highly query-salient
(top-$k$ TF-IDF already achieves ${\sim}80\%$ gold recall).
MuSiQue~\cite{TrivediMultiHop2022} stresses the harder setting: gold passages
span different Wikipedia entities (BM25@5 gold recall only $39\%$), and
48\% of questions require 3--4 supporting passages.
We define a \emph{restricted facility-location} (RFL) function that
diversifies within the query-relevant neighborhood
(Appendix~\ref{app:llm-details}).
On a 50-question pilot, RFL with embedding-based
gating achieves gold-passage recall of $0.740$ at $k=10$,
versus $0.648$ for top-$k$ TF-IDF ($+9.2$ percentage points; bootstrap
95\% CI $[+2.8,+15.5]$ percentage points, $n=50$); per-$k$ breakdowns are in
Appendix~\ref{app:llm-details}.
These experiments suggest that non-monotone submodular pruning is a
viable approach to LLM context selection; a comprehensive study across
tasks and reader models is future work
(downstream F1 and reader saturation in Appendix~\ref{app:reader-saturation}).

\begin{table}[h]
\centering
\small
\begin{tabular}{@{}l ccc cc@{}}
\toprule
Algorithm & F1 & EM & Proxy $\alpha$
          & Gold/$k$ & $|P|$ \\
\midrule
\emph{Gold only} & 0.641 & 0.44 & --- & 2.00 & 2 \\
\emph{All original} & 0.556 & 0.40 & --- & 2.00 & 10 \\
\midrule
StdGreedy ($\omega\!=\!2$) & \textbf{0.599} & \textbf{0.42} & 1.00 & 1.02 & 6 \\
SeqDisjoint ($\omega\!=\!2$) & 0.595 & 0.40 & 1.00 & 1.02 & 6 \\
SeqDisjoint ($\omega\!=\!5$) & 0.509 & 0.36 & 1.00 & 1.02 & 15 \\
WindowMax ($\omega\!=\!5$) & 0.431 & 0.26 & 1.00 & 0.74 & 5.2 \\
QuickPrune & 0.262 & 0.12 & 0.84 & 0.36 & 1.1 \\
Random & 0.274 & 0.16 & 0.79 & 0.38 & 15 \\
\bottomrule
\end{tabular}
\caption{LLM validation of proxy-guided pruning on HotpotQA
($n = 50$, $k = 3$, 50~questions).
Proxy $\alpha$ is measured on $\hat f$; F1/EM are computed by the LLM on the
selected $k$-subsets, and Gold/$k$ is the average number of gold passages
selected.
Proxy-guided pruning at $\omega = 2$ matches or exceeds the full-context baseline
while using only $k = 3$ passages.
The gold-only oracle (F1 $= 0.641 > 0.556$) confirms non-monotonicity of the
true quality function.}
\label{tab:llm-validation}
\end{table}

\section{Conclusion}
\label{sec:conclusion}
In this paper, we make substantial progress towards resolving the complexity of pruning for submodular
maximization. 
Monotone containment is tight at $1-1/e$: greedy achieves it and no algorithm can beat it, regardless of pruning budget (Theorem~\ref{thm:monotone-hardness}).
For non-monotone functions, two independent algorithms achieve $1/2-\varepsilon$ containment under cardinality constraints (Theorems~\ref{thm:half-containment},~\ref{thm:window-containment}), with a knapsack extension (Theorem~\ref{thm:knapsack-half-containment}).
Since $1/2$ exceeds the best non-monotone maximization ratio ($0.401$) and the known hardness threshold ($0.478$), pruning is provably easier than optimization.
Several open questions remain: the gap from $[1/2, 1 - 1/e]$ in the non-monotone setting, and tight algorithms for more sophisticated constraints in the monotone setting. 
Empirically, pruning let an exact IP solver run ${\approx}620\times$ faster on a reduced MaxCut instance; LLM context selection exhibited genuinely non-monotone behavior in the HotpotQA study ($38.7\%$ of questions), and a relevance-gated submodular proxy improved gold-passage recall by $+9.2$ percentage points on a 50-question MuSiQue pilot.

\newpage
\appendix

\section{Additional Related Work}
\label{sec:related}

\paragraph{Ground-set pruning.}
Nath and Kuhnle~\cite{NathKuhnle2025} introduced containment pruning
and gave the first constant-factor guarantee
($\alpha \ge 1/24 - \varepsilon$) for monotone, $\gamma$-weakly submodular
functions under a knapsack constraint via \textsc{QuickPrune}.
Zhou et al.~\cite{ZhouOuyangBilmes2017} gave a heuristic
ground-set reduction via ``pruned submodularity graphs,'' but without
containment guarantees.
Several learned approaches---GCOMB~\cite{Manchanda2020},
LeNSE~\cite{IrelandMontana2022}, and
COMBHelper~\cite{TianEtAl2024}---train pruning policies from data
but likewise lack approximation guarantees; Nath and
Kuhnle showed that \textsc{QuickPrune} empirically dominates them
on the monotone benchmarks studied there.

\paragraph{Bicriteria vs.\ containment.}
The bicriteria framework of Feldman and
Kuhnle~\cite{FeldmanKuhnle2025} outputs $1.44k$ elements containing
a $(1-1/e)$-approximate \emph{solution}.
This differs from containment in three respects:
(i)~bicriteria finds a specific approximate solution, whereas
containment guarantees that near-optimal solutions \emph{exist} in the
pruned set for downstream re-optimization;
(ii)~bicriteria relaxes the budget ($1.44k$ vs.\ strict $k$), whereas
containment enforces $|T| \le k'$;
(iii)~bicriteria targets a single budget, whereas our model requires
containment at all $k' \le k$ simultaneously.
The two models are complementary: bicriteria is stronger when one
approximate solution suffices, containment is stronger when the pruned
set must support multiple downstream tasks or budgets.

\paragraph{Coresets and distributed reduction.}
Composable coresets~\cite{MirrokniZadimoghaddam2015,IndykMahabadiMirrokni2014}
reduce the ground set to $O(k)$ elements for distributed submodular
maximization.
Mirrokni and Zadimoghaddam~\cite{MirrokniZadimoghaddam2015} achieve
$1/3$-approximation for monotone functions by merging local greedy
summaries---this can be viewed as fixed-budget containment.
However, their guarantee holds for a single fixed $k$ and monotone
objectives only; our budget-range formulation requires containment at all
$k' \le k$ simultaneously, our $1/2$ containment handles non-monotone
functions, and our monotone bound $1-1/e$ is strictly stronger than~$1/3$.

\paragraph{Two-stage and streaming approaches.}
Balkanski et al.~\cite{BalkanskiMirzasoleimanSinger2016} study a
two-stage model achieving $\approx (1-1/e)$-approximation,
but their output is a \emph{solution}, not a pruned universe with
containment guarantees.
Streaming algorithms similarly produce approximate solutions rather than
pruned universes amenable to re-optimization:
Sieve-Streaming~\cite{BadanidiyuruMirzasoleiman2014} achieves
$(1/2-\varepsilon)$, and the optimal streaming bound was established by
Alaluf et al.~\cite{AlalufEneFeldman2019}.
Wei et al.~\cite{WeiBilmes2014} use cheap surrogate submodular functions
for multi-round data selection---the ``optimize a proxy, then solve
downstream'' paradigm that our LLM experiments instantiate with the
non-monotone proxy~$\hat f$.

\paragraph{Hardness constructions.}
Our monotone lower bound uses the planted-set construction of
Harshaw et al.~\cite{HarshawFeldmanWardKarbasi2019},
which hides a planted optimal set in the ground set: until the query set
crosses an overlap threshold, the function is indistinguishable from a
null instance with no hidden structure.
This makes planted-set constructions well-suited to pruning hardness,
because any algorithm's output looks random with respect to the hidden
optimum.
The key technical insight for our result is that scaling
$|N| \propto p$ makes the expected overlap between a $p$-element pruned
set and the planted optimum \emph{independent of $p$}, so extra budget
provably does not help.
This differs fundamentally from the \emph{symmetry gap}
technique~\cite{Vondrak2013,GharanVondrak2011,Qi2022}---the standard
tool for submodular maximization hardness---which exploits
group-invariance: the hard instance is a symmetrized function whose value
depends only on a coarse orbit statistic (roughly, how the queried set
distributes across a hidden partition of $N$), so \emph{every}
orbit-balanced subset of sufficient size is near-optimal.
Any pruning algorithm that retains an orbit-balanced slice of $N$---e.g.,
a uniform random $O(k)$-subset---thus achieves containment close to the
maximization ratio, making symmetry-gap constructions unsuitable for
pruning lower bounds.
The curvature-dependent hardness of
Sviridenko et al.~\cite{SviridenkoVondrakWard2017}
inspires Corollary~\ref{cor:curvature-hardness}.

\paragraph{Non-monotone submodular maximization.}
Our containment results sit within the non-monotone approximation
landscape anchored by two numbers: the $0.401$-approximation of
Buchbinder and Feldman~\cite{BuchbinderFeldman2024} (best algorithm)
and the $0.478$-hardness of Qi~\cite{Qi2022} (strongest lower bound).
The $1/2$-hardness of Feige et al.~\cite{FeigeMirrokniVondrak2011}
for unconstrained non-monotone maximization is suggestive---our two
independent $1/2$-containment proofs match this threshold---but no
formal reduction to the containment setting is known.

\paragraph{Submodular optimization for NLP and LLMs.}
Lin and Bilmes~\cite{LinBilmes2011} introduced facility-location
functions for document summarization, establishing the submodular
proxy paradigm.
A growing line of work applies such proxies to LLM pipelines:
Kumari et al.~\cite{KumariWangDasBilmes2024} for in-context example
retrieval, BumbleBee~\cite{KumariWangZhouBilmes2024} for KV-cache
management, DELIFT~\cite{AgarwalKillamsetty2025} for instruction
fine-tuning.
All perform \emph{selection}---choosing a single set of items to
include---using monotone proxies, and none address what guarantees the
pruning stage itself provides.
Our work differs in two ways: we study \emph{pruning} (reducing the
ground set so that downstream solvers can find near-optimal solutions
at any budget), and we introduce the first \emph{non-monotone} proxy
for context selection.

\paragraph{Non-monotonicity in LLM context.}
Liu et al.~\cite{LiuLinHewitt2024} document the ``lost in the middle''
phenomenon---LLMs degrade when relevant information is buried in long
contexts---and Du et al.~\cite{DuEtAl2025} show that context length
alone hurts performance even with perfect retrieval.
These findings motivate our non-monotone proxy
$\hat{f} = \mathrm{FL} - \theta$: the convex penalty $\theta$,
fitted from LLM evaluations via isotonic regression,
models the empirical quality degradation that monotone proxies
cannot capture.

\section{Deferred Preliminaries}
\label{app:deferred-prelim}

\paragraph{BFNS Continuous Double Greedy with warm start.}
For completeness, we state the three-term guarantee from
the Continuous Double Greedy (CDG) of~\cite{BFNS2014}, which we use in
Appendix~\ref{app:warm-start} to analyze a natural warm-start approach
(and show it cannot beat $1/2$):

\begin{theorem}[BFNS Theorem 3.7, adapted]
\label{thm:bfns-cdg}
For a non-negative submodular function $g$ on a ground set of size $n$ with
cardinality constraint $k$, let $v = n/k - 1$ and $s = \sqrt{v}$.
CDG outputs a set $D$ with $|D| \le k$ satisfying
\[
  g(D) \;\ge\; \frac{2s}{(s+1)^2}\,g(\OPT_g)
    \;+\; \frac{v}{(s+1)^2}\,g(\emptyset)
    \;+\; \frac{1}{(s+1)^2}\,g(N).
\]
In particular, at $v = 4$ (density $c = 1/5$), the coefficients on
$g(\OPT_g)$ and $g(\emptyset)$ are both $4/9$.
\end{theorem}

The $g(\emptyset)$ term is normally zero for normalized functions, but becomes
significant when CDG is applied to $g(S) = f(S \cup A)$ for a ``warm start''
set $A$: then $g(\emptyset) = f(A)$, which can be large.

\section{Fast Budget-Range Pruning}
\label{app:fast-monotone}

The monotone hardness shows that $\alpha = 1 - 1/e$ is the
information-theoretic ceiling, and standard greedy~\cite{NemhauserWolsey1978}
attains it at $|P| = k$ in $O(nk)$ value queries by virtue of its
prefix property~\cite{NathKuhnle2025}: the size-$k'$ prefix of greedy's
output achieves $(1-1/e)\,\OPT_{k'}$ for every $k' \le k$.
For very large ground sets the $O(nk)$ query cost is the bottleneck,
so it is natural to ask whether budget-range containment at the full
$(1-1/e-\varepsilon)$ ratio is attainable with near-linear queries.

A simple application of existing single-budget machinery answers this
affirmatively.
Let $\textsc{ThresholdGreedy}(N, f, k',\eta)$ denote Badanidiyuru and
Vondr\'ak's decreasing-threshold
algorithm~\cite{BadanidiyuruVondrak2014}, which for monotone
submodular $f$ returns $S$ with $|S| \le k'$,
$f(S) \ge (1 - 1/e - \eta)\,\OPT_{k'}$, and query complexity
$O((n/\eta) \log(n/\eta))$.

\begin{theorem}[Fast budget-range pruning]
\label{thm:fast-monotone}
Let $f \colon 2^N \to \mathbb{R}_{\ge 0}$ be monotone submodular,
$k \le |N|$, and $\varepsilon \in (0, 1/2)$.
Set $\eta = \varepsilon/4$ and
\[
  \mathcal{K}
  =
  \{1,\ldots,\min\{k,\lceil 1/\eta\rceil\}\}
  \cup
  \{\min\{k,\lceil(1+\eta)^j\rceil\}: j=0,1,\ldots,
     \lceil\log_{1+\eta} k\rceil\}.
\]
let
\[
  P \;=\; \bigcup_{q\in\mathcal{K}} S_q, \qquad
  S_q \;=\; \textsc{ThresholdGreedy}(N, f, q,\eta).
\]
Then $|P| = O(k/\varepsilon)$, the construction uses
$\tilde O(n/\varepsilon^2)$ value queries, and for every $k' \le k$
there exists $Q \subseteq P$ with $|Q| \le k'$ and
\[
  f(Q) \;\ge\; (1 - 1/e - \varepsilon)\,\OPT_{k'}.
\]
\end{theorem}

\paragraph{Proof overview.}
The proof uses two external tools and is otherwise a direct construction.

\smallskip
\noindent\emph{Step~1: Threshold greedy (imported).}
Run Badanidiyuru \& Vondr\'{a}k's decreasing-threshold
algorithm~\cite{BadanidiyuruVondrak2014} at a geometric grid of budgets.
Its guarantee is stated in Fact~\ref{fact:threshold-greedy} below.

\smallskip
\noindent\emph{Step~2: Rounding down (imported).}
When the witness $S_q$ overshoots the target budget $k'$, thin it to a
random $k'$-subset using the standard monotone-submodular
random-subset bound (Fact~\ref{fact:random-subset}).

\smallskip
\noindent\emph{Step~3: Union bound (original).}
The union of all threshold-greedy outputs has total size $O(k/\varepsilon)$
and query cost $\tilde{O}(n/\varepsilon^2)$.

\begin{fact}[ThresholdGreedy guarantee]
\label{fact:threshold-greedy}
The decreasing-threshold greedy algorithm of
Badanidiyuru and Vondr\'{a}k~\cite{BadanidiyuruVondrak2014} has the following
guarantee.
For monotone submodular $f$ with cardinality constraint~$k'$
and accuracy parameter $\eta \in (0,1)$,
$\textsc{ThresholdGreedy}(N, f, k', \eta)$ returns $S$ with
$|S| \le k'$, $f(S) \ge (1-1/e-\eta)\OPT_{k'}$,
using $O((n/\eta)\log(n/\eta))$ value queries.
\end{fact}

\begin{fact}[Random-subset bound]
\label{fact:random-subset}
For monotone submodular functions, the following standard random-subset
bound follows from Nemhauser and Wolsey~\cite{NemhauserWolsey1978}.
Let $f$ be monotone submodular and let $S$ be a set of size $s$.
If $Q$ is a uniformly random $q$-subset of $S$ with $q \le s$, then
$\mathbb{E}[f(Q)] \ge (q/s)\, f(S)$.
\end{fact}

\begin{proof}
\emph{Size.}
The exact small-budget part contributes at most
$O(\min\{k^2,\eta^{-2}\})$, which is $O(k/\eta)$ whether
$k \le 1/\eta$ or $k > 1/\eta$.
The geometric part contributes
$\sum_j O((1+\eta)^j)=O(k/\eta)$.
Thus $|P| \le \sum_{q\in\mathcal{K}}q = O(k/\eta)=O(k/\varepsilon)$.

\emph{Queries.}
There are $O(\eta^{-1}+\log k/\eta)$ calls to
$\textsc{ThresholdGreedy}$, each using $O((n/\eta)\log(n/\eta))$
queries, for a total of $\tilde O(n/\eta^2)=\tilde O(n/\varepsilon^2)$.

\emph{Containment.}
Fix $k' \le k$.
By construction there is a $q\in\mathcal{K}$ with
$k'\le q\le (1+2\eta)k'$: if $k'\le \lceil1/\eta\rceil$ then take
$q=k'$, and otherwise take the first geometric budget at least~$k'$
(the ceiling contributes at most $\eta k'$).
By monotonicity of~$f$, $\OPT_q \ge \OPT_{k'}$, so
Fact~\ref{fact:threshold-greedy} gives
\[
  f(S_q)
  \;\ge\; (1 - 1/e - \eta)\,\OPT_q
  \;\ge\; (1 - 1/e - \eta)\,\OPT_{k'}.
\]
If $|S_q|\le k'$, take $Q=S_q$.
Otherwise, let $Q$ be a uniformly random $k'$-subset of $S_q$.
By Fact~\ref{fact:random-subset} with $s = |S_q|$ and $q = k'$:
\[
  \mathbb{E}[f(Q)]
  \;\ge\; \frac{k'}{|S_q|} f(S_q)
  \;\ge\; \frac{1}{1+2\eta} f(S_q)
  \;\ge\; (1-2\eta)f(S_q).
\]
Therefore
\[
  \mathbb{E}[f(Q)]
  \;\ge\;
  (1-2\eta)(1-1/e-\eta)\,\OPT_{k'}
  \;\ge\; (1-1/e-\varepsilon)\,\OPT_{k'},
\]
where the last inequality uses $\eta=\varepsilon/4$ and
$\varepsilon\in(0,1/2)$.
Hence some realization of the thinning gives the claimed exact-size witness.
\end{proof}

\begin{remark}[Rounding down witnesses]
\label{rem:fast-monotone-open}
The thinning step is used only to certify exact-size containment.
The pruning set itself still has size $O(k/\varepsilon)$; whether one
can achieve the same all-budget guarantee with a core of size $O(k)$
and $o(nk)$ queries remains open.
\end{remark}

\section{Knapsack Proof Sketch and Additional Remarks}
\label{app:knapsack-details}

\begin{proof}[Proof sketch]
This sketches Theorem~\ref{thm:knapsack-half-containment}; the full proof
appears in Appendix~\ref{app:knapsack-proof}.
Fix a query budget $B' \in (0,B]$ and rerun the analysis for this budget.
Let $O$ be optimal for $B'$, let $\hatO=O\setminus P$ and
$\dotO=O\cap P$, and set $\beta=c(\hatO)\le B'$.
If $\beta=0$, then $O\subseteq P$ and $Q=O$ proves the claim.
Assume $\beta>0$.

For each density-greedy run~$i$, take the prefix whose cost is exactly
$\beta$.  If this prefix cuts through original item~$a_{j_i}^{(i)}$, apply
the split operation of Definition~\ref{def:split-appendix}: write
$\tau_i$ for the included fragment of~$a_{j_i}^{(i)}$ and let the other
fragment remain excluded.  Thus $A'_i$ is a set in the split ground set with
$c(A'_i)=\beta$; it is either a real prefix or a real prefix plus the
single included boundary fragment~$\tau_i$.
Let $f^*$ be the resulting split extension, and write $S^\uparrow$ for the
full lift of an original set~$S$.

The sets $A'_i$ are disjoint, so the same averaging step as in
Theorem~\ref{thm:half-containment} gives a run~$i^\star$ such that
\[
  f^*(O^\uparrow\cup A'_{i^\star}) \ge (1-1/\ell)f(O).
\]
Set $A=A'_{i^\star}$.  Since every element of $\hatO$ survives all previous
runs and is therefore a candidate throughout run~$i^\star$, the
density-domination lemma (Lemma~\ref{lem:density-domination}) gives
\[
  f^*(\hatO^\uparrow\cup A) \le 2f^*(A)-f(\emptyset).
\]
Submodularity applied to
$\hatO^\uparrow\cup A$ and $\dotO^\uparrow\cup A$ then implies
\[
  \max\{f^*(A),\,f^*(\dotO^\uparrow\cup A)\}
  \ge (1/2-\varepsilon)\OPT_{B'}.
\]

To report a real feasible subset, delete the included boundary fragment
from $A$ if one exists; call the resulting original set $\widehat A$.
The candidates $\widehat A$ and $\dotO\cup\widehat A$ both have cost at most
$B'$.  If the selected run has boundary fragment~$\tau_{i^\star}$, let
$a_{j^\star}^{(i^\star)}$ be the original item containing it and let
$\gamma_{i^\star}$ be the included fraction of that item.  Deleting the
fragment loses at most
$\Delta_{B'}=\gamma_{i^\star}f(\{a_{j^\star}^{(i^\star)}\})$; if no split
occurs, $\Delta_{B'}=0$.  In the split case, the item
$a_{j^\star}^{(i^\star)}$ is budget-specific, but it may still have
$p_{a_{j^\star}^{(i^\star)}}>B'$, so this loss cannot in general be charged
to the best feasible singleton at budget~$B'$.
\end{proof}

\begin{remark}[Simultaneous guarantee across budgets]
\label{rem:budget-range}
The pruned set $P$ depends only on the master budget $B$, not on $B'$;
for each $B' \in (0,B]$, rerun the proof with that budget.  The values
$\beta$, the boundary fragments, and the averaged run $i^\star$ may all change
with~$B'$, but the underlying pruned set~$P$ is fixed.
This extends the budget-range containment model of~\cite{NathKuhnle2025}
(cardinality) to knapsack, with the explicit boundary-fragment loss in
Theorem~\ref{thm:knapsack-half-containment}.
The item-splitting machinery is imported from
Feldman and Kuhnle~\cite{FeldmanKuhnle2025}, originally developed for their
bicriteria knapsack algorithm; we restate the needed pieces
self-contained in Appendix~\ref{app:item-splitting} and adapt the
density-domination invariant to the source/target pair
$(\hatO, A)$ rather than their $(E_2, E_1)$.
\end{remark}

\begin{remark}[Knapsack hardness inherited from cardinality]
\label{rem:knapsack-hardness}
The monotone hardness of Theorem~\ref{thm:monotone-hardness}
($1-1/e$ for every pruning budget) transfers verbatim to knapsack:
a unit-cost instance is a cardinality instance, so any knapsack
pruning algorithm that beat the $1-1/e$ barrier would also beat
it under cardinality---a contradiction.
The same argument shows the non-monotone upper bound $\le 1-1/e$
holds for knapsack, placing all three settings in Table~\ref{tab:summary}
within the gap $[1/2,\, 1-1/e]$.
\end{remark}

\begin{remark}[Why the boundary term is explicit]
\label{rem:knapsack-simple}
A tempting shortcut is to drop the boundary fragment and compare the lost
value to the singleton value of the boundary item. This gives a clean
singleton-style guarantee for a \emph{fixed} budget when the boundary item is
itself feasible for that budget. For a query budget~$B'$, however, the
boundary item is simply the item that crosses the residual cost in the
density order for the proof at~$B'$. It can have cost larger than~$B'$ even
though the included fragment has cost at most the residual budget, so the
original item is not a feasible singleton witness. For this reason
Theorem~\ref{thm:knapsack-half-containment}
keeps the boundary-fragment loss explicit instead of hiding it inside a
best-feasible-singleton term.
\end{remark}

\section{Full Proof of Theorem~\ref{thm:monotone-hardness}}
\label{app:monotone-proof}

\begin{theorem*}[Theorem~\ref{thm:monotone-hardness}, restated]
For every $\gamma \in (0, 1]$, every $\varepsilon > 0$,
and every polynomial $C(\cdot)$,
there exists $k_0 = k_0(\gamma, \varepsilon, C)$ such that for all $k \ge k_0$
and \emph{every} $p \ge k$:
no (possibly randomized) algorithm making at most $C(|N|)$ value queries
and outputting $P \subseteq N$ with $|P| \le p$ can guarantee
\[
  \mathbb{E}\!\left[\max_{S \subseteq P,\, |S| \le k} f(S)\right]
  > (1 - e^{-\gamma} + \varepsilon) \cdot \OPT_k
\]
for all monotone $\gamma$-weakly DR $f\colon 2^N \to \mathbb{R}_{\ge 0}$
with $|N| = \lceil 60p/\varepsilon \rceil$.
At $\gamma = 1$ (full submodularity), this recovers the tight bound
$\alpha \le 1 - 1/e$.
\end{theorem*}

The proof builds on the hard-instance family of Harshaw, Feldman, Ward
\& Karbasi~\cite{HarshawFeldmanWardKarbasi2019} (hereafter HFWK).
We first state the construction in the form used here, then give a proof
overview that separates the HFWK import from our adaptations.

\begin{proposition*}[HFWK planted family, adapted form]
Fix $\gamma \in (0,1]$, $\varepsilon' \in (0,1/6)$, and integers
$k \ge 1/\varepsilon'$ and $m \ge \lceil 3k/\varepsilon' \rceil$.
Let $\N$ be a ground set of size~$m$ and set
$g = \lceil \varepsilon'k + 3k^2/m \rceil$.
For every $k$-subset $T \subseteq \N$, the HFWK construction gives
functions $f_T\colon 2^\N \to \mathbb{R}_{\ge 0}$ that depend on~$T$
through the statistic
\[
  t_T(S) = |S \setminus T| + \min\{g, |S \cap T|\}.
\]
The resulting family $\{f_T\}_{T \in \binom{\N}{k}}$ satisfies:
\begin{enumerate}
  \item[\textup{(P1)}] $f_T$ is non-negative, monotone, and
    $\gamma$-weakly diminishing returns;
  \item[\textup{(P2)}] $f_T(S) \le 1$ for all $S \subseteq \N$,
    with equality when $S = T$;
  \item[\textup{(P3)}] $f_\varnothing(S) \le 1 - e^{-\gamma} + 12\varepsilon'$
    for all $S$ with $|S| \le k$;
  \item[\textup{(P4)}] $f_T(S) = f_\varnothing(S)$ whenever
    $|S \cap T| \le g$ or $|S| \ge 3k - g$.
\end{enumerate}
The case $m = \lceil 3k/\varepsilon' \rceil$ is
HFWK Proposition~14.  The extension to any larger ground set is the
adaptation used here: inspecting their proof shows that
properties~\textup{(P1)--(P4)} use only $m \ge 3k/\varepsilon'$; no step
requires equality.
\end{proposition*}

\paragraph{Proof overview.}
The proof of Theorem~\ref{thm:monotone-hardness} proceeds in three steps.

\smallskip
\noindent\emph{Setup.}
We import the HFWK planted family and properties~(P1)--(P4), but instantiate it
on an \emph{enlarged} ground set with
$m=|\N|=\lceil 3p/\varepsilon'\rceil$ instead of
$m=\lceil 3k/\varepsilon'\rceil$.
The enlargement is our adaptation: it ensures that
$|P|\le p\le \varepsilon' m/3$, making the $p$-cancellation below possible.

\smallskip
\noindent\emph{Indistinguishability.}
We show that a polynomial-query algorithm cannot distinguish
$f_{\tilde{T}}$ from the null instance $f_\varnothing$ with high probability,
using a \emph{two-regime} hypergeometric tail argument
(Fact~\ref{fact:hyp-conc} below) that covers both $p = O(k)$, where the
additive tail suffices, and $p \gg k$, where the multiplicative tail is needed.
This two-regime analysis is new; HFWK's original argument handles only
$m = \Theta(k/\varepsilon')$.

\smallskip
\noindent\emph{Overlap bound.}
Since the algorithm's output $P$ is fixed by the null transcript, the expected
overlap satisfies
\[
  \mathbb{E}[|P \cap \tilde{T}|] \le \frac{pk}{m}
  \le \frac{\varepsilon' k}{3}.
\]
The upper bound is \emph{independent of~$p$} because the $p$ in the pruning
budget cancels with the $p$ in $m=|\N|\propto p$.
Concentration shows the overlap stays below the threshold $g$, so
property~(P4) forces all $k$-subsets of $P$ to match the null-instance value,
bounded by (P3).

\noindent The $p$-cancellation mechanism is the central insight: enlarging
the ground set in proportion to the pruning budget neutralizes the
additional budget, yielding a hardness bound that holds uniformly for
all $p \ge k$. We state the concentration tools used throughout,
then give the full proof.

\begin{fact}[Hypergeometric concentration~{\cite{Hoeffding1963}}]
\label{fact:hyp-conc}
Let $X \sim \mathrm{Hypergeometric}(m, K, n)$, i.e., $X$ counts
marked items when drawing $n$ items without replacement from a population
of~$m$ containing $K$ marked items.  Write $\mu = nK/m$.
\begin{enumerate}
\item[\textup{(i)}]
  \textup{(Additive tail.)}
  For every $t > 0$,
  $\;\Pr[X - \mu \ge t] \le \exp(-2t^2/n)$.
\item[\textup{(ii)}]
  \textup{(Multiplicative tail.)}
  For every integer $a > \mu$,
  $\;\Pr[X \ge a] \le (e\mu/a)^{a}$.
\end{enumerate}
Both follow from Hoeffding's Theorem~4, which shows that hypergeometric
tails are dominated by binomial tails; bound~\textup{(i)} is then the
standard sub-Gaussian tail, and~\textup{(ii)} is the Chernoff--Poisson
bound applied to the dominating binomial.
\end{fact}

\medskip

Fix $\gamma \in (0,1]$, $p \geq k$, and $\varepsilon > 0$.
We may assume $\varepsilon\le1$; larger $\varepsilon$ only weakens the
target bound.  Set $\varepsilon' = \varepsilon/20$, so
$\varepsilon'\in(0,1/6)$, and take $k$ large enough that
$k\ge1/\varepsilon'$, as required by the HFWK construction; this lower bound
is absorbed into the theorem's constant~$k_0$.
Set $m=\lceil 3p/\varepsilon'\rceil=\lceil 60p/\varepsilon\rceil$.
We instantiate the adapted HFWK family stated above on a
ground set $\N$ of size~$m$.  Since
$m \ge 3p/\varepsilon' \ge 3k/\varepsilon'$, all properties
(P1)--(P4) hold for this instance.
The only quantity that changes relative to HFWK's original setting is
the overlap threshold
\[
  g=\left\lceil \varepsilon'k+\frac{3k^2}{m}\right\rceil
  \le \left\lceil \varepsilon'k\left(1+\frac{k}{p}\right)\right\rceil
  \le 2\varepsilon'k+1,
\]
where the last inequality uses $k/p\le1$.
By Yao's minimax principle, it suffices to show the bound for
deterministic algorithms against the uniform distribution over
$\{f_{\tilde{T}}\}$, where $\tilde{T}$ is a random $k$-subset of~$\N$.
Let $\textsc{Alg}$ be a deterministic algorithm
making at most $C(|\N|)$ queries.
Write $C(x)\le A x^d$ for constants $A,d$ and all $x\ge1$.

\medskip\noindent
\emph{Step 1: Indistinguishability.}\;
\textbf{Goal:} show that the algorithm cannot tell whether it is facing
the null instance $f_\varnothing$ (which has no planted optimum) or the
real instance $f_{\tilde{T}}$ (which does).
If the algorithm cannot distinguish the two, its output $P_0$ is the same in
both cases, and any analysis under the null instance applies to the real one.

When given $f_\varnothing$, $\textsc{Alg}$ makes queries
$S_1, \ldots, S_\ell$ ($\ell \le C(|\N|)$) and outputs
$P_0$ with $|P_0| \le p$.
Fix any query set~$S$ on the null-instance transcript.
If $|S|\ge 3k-g$ (a large query), then $f_{\tilde T}(S)=f_\varnothing(S)$
for every $\tilde T$ by~(P4)---the query is too large to reveal the planted set.
Otherwise $|S|<3k-g$, so the query can differ only if the overlap
$X=|S\cap\tilde T|$ exceeds the threshold~$g$.
The random variable $X$ is hypergeometric with mean $\mu=k|S|/m$.
Since $|S|<3k-g$ and
$g\ge \varepsilon'k+3k^2/m$, we have
$g/k-|S|/m\ge(\varepsilon'+3k/m)-(3k-g)/m\ge\varepsilon'$,
so the gap between~$g$ and the mean satisfies $g - \mu \ge \varepsilon' k$.
Fact~\ref{fact:hyp-conc}(i) with $t = \varepsilon' k$ and $n = k$ gives
\[
  \Pr[X>g]
  \le
  \exp(-2\varepsilon'^2 k).
\]
We also apply
Fact~\ref{fact:hyp-conc}(ii) with $a=\lceil\varepsilon'k\rceil$:
since $|S|<3k$ and $m\ge3p/\varepsilon'$,
\[
  \Pr[X>g]
  \le
  \Pr[X\ge \varepsilon'k]
  \le
  \left(\frac{e k}{p}\right)^{\varepsilon'k}.
\]
The two bounds complement each other across all pruning budgets.
When $p$ is small (at most $2ek$), the ground set is only $O(k/\varepsilon')$ and the additive tail decays fast enough:
$C(|\N|)\exp(-2\varepsilon'^2 k)\le\varepsilon/40$ for sufficiently large~$k$.
When $p$ is large ($p>2ek$), the multiplicative tail dominates because $ek/p<1/2$:
\[
  C(|\N|)\left(\frac{e k}{p}\right)^{\varepsilon'k}
  \le
  A\left(\frac{4p}{\varepsilon'}\right)^d
  \left(\frac{e k}{p}\right)^{\varepsilon'k}
  \le \varepsilon/40
\]
for sufficiently large~$k$, uniformly over all $p>2e k$: writing
$r=p/k>2e$, the middle expression is
$A(4/\varepsilon')^d k^d r^d(e/r)^{\varepsilon'k}$, whose maximum over
$r>2e$ occurs at the boundary once $\varepsilon'k>d$ and then decays
exponentially in~$k$.
By a union bound over the at most $C(|\N|)$ null-transcript queries, the event
$\mathcal{E} = \bigl\{f_\varnothing(S_i) = f_{\tilde{T}}(S_i)
  \text{ for all } i = 1, \ldots, \ell\bigr\}$
satisfies
$\Pr[\mathcal{E}]\ge 1-\varepsilon/40$ for sufficiently large~$k$.
Under $\mathcal{E}$, every query returns the same value as under~$f_\varnothing$,
so $\textsc{Alg}$ follows the same execution path and outputs~$P_0$.

\medskip\noindent
\emph{Step 2: Overlap concentration.}\;
\textbf{Goal:} show that the algorithm's output $P_0$ has too little overlap
with the planted optimum $\tilde{T}$ to contain a good solution.

Since Step~1 established that $P_0$ is determined by the null transcript
(independent of $\tilde{T}$), the overlap $|P_0 \cap \tilde{T}|$ is hypergeometric with
\[
  \mathbb{E}\bigl[|P_0 \cap \tilde{T}|\bigr]
  = \frac{k \cdot |P_0|}{|\N|}
  \le \frac{k \cdot p}{3p/\varepsilon'}
  = \frac{\varepsilon' k}{3}.
\]
Note that $p$ cancels: the expected overlap is \emph{independent} of the
pruning budget.
From the definition $g = \lceil \varepsilon'k + 3k^2/|\N| \rceil \ge \varepsilon'k$
(the ceiling dominates the first summand),
the gap satisfies
$g - \mathbb{E}[|P_0 \cap \tilde{T}|] \ge 2\varepsilon'k/3$.
By Fact~\ref{fact:hyp-conc}(i) with $n = k$ and
$t = 2\varepsilon'k/3$:
\[
  \Pr\bigl[|P_0 \cap \tilde{T}| > g\bigr]
  \;\le\; e^{-2t^2/k}
  \;=\; e^{-8\varepsilon'^2 k/9}
  \;=\; e^{-\Omega(\varepsilon^2 k)}
  \;\le\; \varepsilon/20
\]
for sufficiently large~$k$.

\medskip\noindent
\emph{Step 3: Containment bound.}\;
\textbf{Goal:} combine Steps~1 and~2 to bound the expected containment.

Define the good event
$\mathcal{G} = \mathcal{E} \cap \{|P_0 \cap \tilde{T}| \le g\}$.
By union bound, $\Pr[\mathcal{G}] \ge 1 - \varepsilon/10$.
Under~$\mathcal{G}$:
\begin{itemize}
\item $\textsc{Alg}$ outputs $P_0$ (by $\mathcal{E}$);
\item every $S \subseteq P_0$ with $|S| \le k$ satisfies
  $|S \cap \tilde{T}| \le |P_0 \cap \tilde{T}| \le g$,
  so $f_{\tilde{T}}(S) = f_\varnothing(S)$ by~(P4);
\item $\max_{S \subseteq P_0,\, |S| \le k} f_{\tilde{T}}(S)
  = \max_{S \subseteq P_0,\, |S| \le k} f_\varnothing(S)
  \le 1 - e^{-\gamma} + 12\varepsilon'$ by~(P3).
\end{itemize}
Since $\OPT = f_{\tilde{T}}(\tilde{T}) = 1$ by~(P2)
and $f_{\tilde{T}}(S) \le 1$ always:
\begin{align*}
  \mathbb{E}\!\left[\max_{S \subseteq P,\, |S| \le k}
    f_{\tilde{T}}(S)\right]
  &\le \Pr[\mathcal{G}] \cdot (1 - e^{-\gamma} + 12\varepsilon')
    + \Pr[\neg\mathcal{G}] \cdot 1 \\
  &\le (1 - e^{-\gamma} + 12\varepsilon') + \varepsilon/10
  \;=\; 1 - e^{-\gamma} + \tfrac{7\varepsilon}{10}
  \;<\; 1 - e^{-\gamma} + \varepsilon. \qedhere
\end{align*}

\section{Proof of Corollary~\ref{cor:curvature-hardness}}
\label{app:curvature-proof}

\begin{corollary*}[Corollary~\ref{cor:curvature-hardness}, restated]
Let $\kappa \in (0, 1]$ denote total curvature:
$\kappa = 1 - \min_{j \in N} f(j \mid N \setminus \{j\})/f(j)$.
For every $p \ge k$ and every $\varepsilon > 0$,
no algorithm using polynomially many value queries can prune to
$p$ elements with containment
$\alpha > 1 - \kappa/e + \varepsilon$
for all monotone submodular functions with total curvature at most~$\kappa$.
Since the algorithm of Sviridenko, Vondr\'{a}k \&
Ward~\cite{SviridenkoVondrakWard2017} achieves
$\alpha = 1 - \kappa/e$, this bound is tight.
\textup{(}At $\kappa = 0$ the function is modular, so exact containment is trivial.\textup{)}
\end{corollary*}

\paragraph{Proof overview.}
The proof is a two-step reduction.

\smallskip
\noindent\emph{Step~1: SVW augmentation (imported).}
Sviridenko, Vondr\'{a}k \& Ward~\cite{SviridenkoVondrakWard2017}
show that augmenting the normalized planted-optimum family used in their
curvature lower bound with a linear term produces controlled curvature
without altering the information-theoretic landscape
(Fact~\ref{fact:svw-curvature} below).

\smallskip
\noindent\emph{Step~2: Transfer from Theorem~\ref{thm:monotone-hardness} (original).}
Because the linear term depends only on $|A|$, it carries no information
about the planted set~$T$; the indistinguishability and overlap arguments
from Theorem~\ref{thm:monotone-hardness} transfer directly.

\begin{fact}[SVW curvature augmentation]
\label{fact:svw-curvature}
In the normalized form used by Sviridenko, Vondr\'{a}k, and
Ward~\cite{SviridenkoVondrakWard2017}, the curvature augmentation has the
following form.
Let $\{f_T\}$ be a planted-optimum family on ground set $\N$ in the
normalization used by the SVW lower bound, with
$f_T(T) = k\,p_{\max}$ and $p_{\max} = \max_i f_T(\{i\})$.
For $\kappa \in (0,1]$, define
$\hat{f}_T(A) = f_T(A) + \frac{1-\kappa}{\kappa}\,|A|\,p_{\max}$.
Then: \textup{(i)}~$\hat{f}_T$ is monotone submodular with total curvature
$\le \kappa$; \textup{(ii)}~$\hat{f}_T(T) = k\,p_{\max}/\kappa$;
\textup{(iii)}~the linear term is determined by $|A|$ alone and therefore
provides no information about~$T$.
\end{fact}

\begin{proof}[Proof of Corollary~\ref{cor:curvature-hardness}]

\emph{Step~1 (SVW augmentation).}
Apply Fact~\ref{fact:svw-curvature} to the HFWK family from
Theorem~\ref{thm:monotone-hardness} with the same enlarged ground set
$m = \lceil 3p/\varepsilon' \rceil$.
The resulting $\hat{f}_T$ has curvature $\le \kappa$.

\emph{Step~2 (Transfer from Theorem~\ref{thm:monotone-hardness}).}
By Fact~\ref{fact:svw-curvature}(iii), the linear term carries no
information about the planted set.  All queries on $\hat{f}_T$ reduce
to queries on $f_T$ plus a cardinality lookup, so the
indistinguishability and overlap concentration arguments from
Theorem~\ref{thm:monotone-hardness} apply unchanged.
Under the good event (overlap below $g$), every $k$-subset
$S \subseteq P$ satisfies
\[
  \hat{f}_T(S)
  \;\le\; (1 - 1/e + O(\varepsilon))\,k\,p_{\max}
    + \tfrac{1-\kappa}{\kappa}\,k\,p_{\max}
  \;=\; \tfrac{1 - \kappa/e + O(\varepsilon)}{\kappa}\,k\,p_{\max},
\]
while $\hat{f}_T(T) = k\,p_{\max}/\kappa$
by Fact~\ref{fact:svw-curvature}(ii).
The ratio is $1 - \kappa/e + O(\varepsilon)$.
\end{proof}

\begin{table}[h]
\centering
\begin{tabular}{@{}lccccc@{}}
\toprule
$p$ & $c' = k/p$ & Composition $\alpha \le$ & Hardness $\alpha \le$ & Best algo $\alpha$ & Status \\
\midrule
$k$ & $1$ & $1-1/e$ & $1-1/e$ & $1-1/e$ & \textbf{Tight} \\
$2k$ & $1/2$ & $0.843$ & $1-1/e$ & $1-1/e$ & \textbf{Tight} \\
$ck$ (fixed $c>1$) & $1/c$ & $(1-1/e)/\alpha^*(1/c)$ & $1-1/e$ & $1-1/e$ & \textbf{Tight} \\
$\omega(k)$ & $\to 0$ & $\to 1$ & $1-1/e$ & $1-1/e$ & \textbf{Tight} \\
\bottomrule
\end{tabular}
\caption{Monotone containment landscape at $\gamma = 1$ (full submodularity).
The composition column is the benchmark obtained by ``prune then optimize'';
it matches $1-1/e$ only at $p=k$ and weakens for larger pruning budgets
(approaching~$1$ only when $p/k$ grows).
The hardness and best-algorithm columns show that the true monotone containment factor remains exactly $1-1/e$ for every $p\ge k$.
For $\gamma$-weakly DR functions, replace $1-1/e$ with $1-e^{-\gamma}$; for curvature~$\kappa$, use $1-\kappa/e$ (Corollary~\ref{cor:curvature-hardness}).}
\label{tab:monotone}
\end{table}

\section{Full Proof of Theorem~\ref{thm:half-containment}}
\label{app:half-containment-proof}

\begin{theorem*}[Theorem~\ref{thm:half-containment}, restated]
Let $f$ be a non-negative submodular function on ground set $N$, and let
$k \le n$.
For any $\varepsilon > 0$, there exists a set $G \subseteq N$ with
$|G| = O(k/\varepsilon)$ such that for every $k' \le k$,
\[
  \max_{T \subseteq G,\, |T| \le k'} f(T)
  \;\ge\; \left(\tfrac{1}{2} - \varepsilon\right) \cdot \OPT_{k'},
\]
where $\OPT_{k'} = \max_{|S| \le k'} f(S)$.
Moreover, $G$ can be constructed in $O(nk/\varepsilon)$ oracle queries.
\end{theorem*}

\paragraph{Proof overview.}
The proof uses the \emph{sequential disjoint greedy} (SDG) scheme and
proceeds in four steps, all original to this paper.
The high-level idea: run greedy many times on disjoint parts of the ground set, find the best run by averaging, then show that this run's greedy prefix captures at least half the optimal value.

\smallskip
\noindent\emph{Step~1: Averaging.}
Build $\ell$ disjoint greedy solutions.
Because the runs use disjoint elements, their total ``damage'' to the optimal value is bounded by submodularity (Fact~\ref{fact:submod-disjoint}).
Averaging selects the run that complements OPT best, losing only a $(1-1/\ell)$ factor.

\smallskip
\noindent\emph{Step~2: Partial-greedy domination.}
The elements of $\hatO = O \setminus G$ (the ``missed'' OPT elements) were available as candidates during the chosen greedy run but were not picked.
Since greedy always picks the highest-marginal element, each greedy pick dominates the corresponding missed element.
Telescoping gives $f(\hatO \cup A) \le 2f(A) - f(\emptyset)$: the greedy prefix captures at least half the value of the prefix augmented with the missed elements.

\smallskip
\noindent\emph{Step~3: Submodularity decomposition.}
The lattice inequality $f(X \cup Y) + f(X \cap Y) \le f(X) + f(Y)$,
applied to $X = \hatO \cup A$ and $Y = \dotO \cup A$,
relates the value of $O \cup A$ to the two feasible candidates $A$ and $\dotO \cup A$.

\smallskip
\noindent\emph{Step~4: Feasibility.}
Both $A$ and $\dotO \cup A$ have size at most $k'$ and lie in $G$, so the maximum over them
gives the claimed $1/2 - \varepsilon$ containment.

\noindent The argument uses one standard submodularity fact for disjoint sets,
stated for reference:

\begin{fact}[Submodularity of disjoint-union marginals]
\label{fact:submod-disjoint}
Let $f$ be non-negative and submodular, and let
$A_1, \ldots, A_\ell$ be pairwise disjoint.
Then for any set~$O$:
\begin{align*}
  \sum_{i=1}^\ell f(A_i \mid O)
  &\ge
  f\!\bigl(\textstyle\bigcup_{i} A_i \,\big|\, O\bigr),\\
  f\!\bigl(\textstyle\bigcup_{i} A_i \,\big|\, O\bigr)
  &\le
  f\!\bigl(\textstyle\bigcup_{i} A_i \cup O\bigr).
\end{align*}
This follows from repeated application of the diminishing-returns property
of submodularity: adding disjoint sets one at a time, each marginal is at
least as large as the joint marginal.
\end{fact}

\begin{proof}[Proof of Theorem~\ref{thm:half-containment}]
Set $\ell = \lceil 1/\varepsilon \rceil$.
If $|N| < \ell k$, output $G = N$; containment is exact since every $k'$-subset
of~$N$ is available, and $|G| = |N| < \ell k = O(k/\varepsilon)$.
Otherwise, build $\ell$ disjoint greedy solutions sequentially:
$G_1 = \mathrm{greedy}(N, k)$,
$G_2 = \mathrm{greedy}(N \setminus G_1, k)$, \ldots,
$G_\ell = \mathrm{greedy}(N \setminus \bigcup_{j<\ell} G_j, k)$.
Let $G = \bigcup_{i=1}^\ell G_i$, so $|G| = \ell k$.

Fix any $k' \le k$.
Write $O$ for an optimal feasible set: $|O| \le k'$ and
$f(O) = \OPT_{k'}$.
Let $b = |O \setminus G|$ and decompose:
\[
  \dotO = O \cap G, \quad \hatO = O \setminus G, \quad
  |\dotO| = |O| - b, \quad |\hatO| = b.
\]

\medskip\noindent\textbf{Step 1 (Averaging).}
For each $i$, let $A_i$ denote the first $b$ elements of greedy run~$i$.
Since $A_i \subseteq G_i$ and the $G_i$ are disjoint, the $A_i$ are disjoint.
By Fact~\ref{fact:submod-disjoint} applied to the disjoint collection:
\[
  \sum_{i=1}^\ell \bigl[f(O) - f(O \cup A_i)\bigr]
  = -\sum_{i=1}^\ell f(A_i \mid O)
  \le -f\!\bigl(\textstyle\bigcup A_i \mid O\bigr)
  = f(O) - f\!\bigl(O \cup \textstyle\bigcup A_i\bigr) \le f(O).
\]
By averaging, there exists $i^*$ with
$f(O \cup A_{i^*}) \ge (1 - 1/\ell)\, f(O)$.
Set $A = A_{i^*}$.

\medskip\noindent\textbf{Step 2 (Partial greedy).}
Recall $A = A_{i^*}$, the first $b$ elements of greedy run~$i^*$
(note $b \le k' \le k$, so these elements are well-defined).
Since $\hatO \cap G_j = \emptyset$ for every $j$ (by definition of $\hatO = O \setminus G$),
every element of $\hatO$ is available as a candidate in greedy run~$i^*$.
Thus each greedy pick $a_j$ satisfies
$f(a_j \mid A_{j-1}) \ge f(\hat{o}_j \mid A_{j-1})$
for some ordering $\hat{o}_1, \ldots, \hat{o}_b$ of $\hatO$
(since greedy picks the maximum marginal and $\hat{o}_j$ is available).

By the telescoping chain and submodularity:
\begin{align}
  f(A) - f(\emptyset) &= \sum_{j=1}^b f(a_j \mid A_{j-1})
  \ge \sum_{j=1}^b f(\hat{o}_j \mid A_{j-1})
  \ge \sum_{j=1}^b f(\hat{o}_j \mid A) \notag \\
  &\ge f(\hatO \mid A) = f(\hatO \cup A) - f(A). \label{eq:greedy-half-app}
\end{align}
Therefore $2f(A) - f(\emptyset) \ge f(\hatO \cup A)$.

\medskip\noindent\textbf{Step 3 (Submodularity decomposition).}
Apply submodularity with $X = \hatO \cup A$ and $Y = \dotO \cup A$:
$f(X \cup Y) + f(X \cap Y) \le f(X) + f(Y)$.
Since $X \cup Y = O \cup A$ and $X \cap Y = A$:
\[
  f(O \cup A) + f(A) \le f(\hatO \cup A) + f(\dotO \cup A).
\]
Using~\eqref{eq:greedy-half-app}: $f(\hatO \cup A) \le 2f(A) - f(\emptyset)$, so
\begin{equation}
  f(O \cup A) \le f(A) - f(\emptyset) + f(\dotO \cup A). \label{eq:decomposition-app}
\end{equation}

\medskip\noindent\textbf{Step 4 (Feasibility and conclusion).}
Both candidates are feasible subsets of $G$:
$A \subseteq G_{i^*} \subseteq G$ with $|A| = b \le k'$;
and $\dotO \cup A \subseteq G$ with
$|\dotO \cup A| \le |\dotO| + |A| = |O| \le k'$.
From~\eqref{eq:decomposition-app} and Step~1:
\[
  \max\bigl(f(A),\, f(\dotO \cup A)\bigr)
  \ge \frac{f(O \cup A) + f(\emptyset)}{2}
  \ge \frac{(1 - 1/\ell) \cdot \OPT_{k'} + f(\emptyset)}{2}.
\]
Since $f(\emptyset) \ge 0$ and $k' \le k$ was arbitrary with $\ell = \lceil 1/\varepsilon \rceil$,
the set $G$ achieves containment factor
$\alpha \ge (1-1/\ell)/2 \ge 1/2 - \varepsilon$ simultaneously at every
scale $k' \le k$.
\end{proof}

\section{Item-Splitting Machinery and Proof of Theorem~\ref{thm:knapsack-half-containment}}
\label{app:item-splitting}

This appendix makes the knapsack extension fully self-contained.

\emph{Why item splitting is needed.}
In the cardinality setting, each greedy prefix of length $b$ matches the budget exactly.
With knapsack constraints, a density-greedy prefix may overshoot the target budget:
the last item added may have cost larger than the remaining budget allows.
Item splitting resolves this by conceptually breaking the boundary item into two
fragments---one that fits the budget and one that does not---while preserving
submodularity of the objective function.

The item-splitting construction is due to Feldman and
Kuhnle~\cite{FeldmanKuhnle2025} (Definitions~7.13 and~7.15 there);
we restate the pieces needed for Theorem~\ref{thm:knapsack-half-containment}
and defer the submodularity-preservation proof
to~\cite{FeldmanKuhnle2025}.

\subsection{The \texorpdfstring{$\gamma$}{gamma}-split extension}

\begin{definition}[Gamma-split]
\label{def:split-appendix}
This is the split operation of Feldman and Kuhnle~\cite{FeldmanKuhnle2025}.
Let $f\colon 2^N \to \mathbb{R}_{\ge 0}$ be submodular with positive price
vector $p$. For $s \in N$ and $\gamma \in (0, 1)$, the \emph{$\gamma$-split
of $s$} produces a new ground set $N^* = (N \setminus \{s\}) \cup
\{\sigma_1, \sigma_2\}$, prices $p^*_{\sigma_1} = \gamma p_s$,
$p^*_{\sigma_2} = (1 - \gamma) p_s$, $p^*_u = p_u$ otherwise, and
function
\[
  f^*(S) \;=\; f(\widehat S)
  \;+\; \gamma \cdot f(s \mid \widehat S) \cdot \mathbb{I}[\sigma_1 \in S]
  \;+\; (1 - \gamma) \cdot f(s \mid \widehat S) \cdot \mathbb{I}[\sigma_2 \in S],
\]
where $\widehat S = S \setminus \{\sigma_1, \sigma_2\}$. When
$S \cap \{\sigma_1, \sigma_2\} = \emptyset$, $f^*(S) = f(S)$. When
$\{\sigma_1, \sigma_2\} \subseteq S$, $f^*(S) = f(\widehat S + s)$.
\end{definition}

\begin{lemma}[Split preserves submodularity]
\label{lem:split-submod-appendix}
Feldman and Kuhnle~\cite[Lemma~7.14]{FeldmanKuhnle2025} prove that
the function $f^*$ of Definition~\ref{def:split-appendix} is non-negative
and submodular on $N^*$. Iterating the split operation finitely many times
preserves both properties.
\end{lemma}

\begin{observation}[Density preservation]
\label{obs:density-preservation}
This is the density-preservation property from
Feldman and Kuhnle~\cite[Obs.~7.16]{FeldmanKuhnle2025}.
After a $\gamma$-split of $s$, for any set $S \subseteq N^*$ with
$\sigma_1 \notin S$ resp.\ $\sigma_2 \notin S$,
\[
  f^*(\sigma_1 \mid S) \;=\; \gamma \cdot f(s \mid \widehat S),
  \qquad
  f^*(\sigma_2 \mid S) \;=\; (1 - \gamma) \cdot f(s \mid \widehat S).
\]
Together with $p^*_{\sigma_i} = \gamma_i p_s$ (with $\gamma_1 = \gamma$,
$\gamma_2 = 1 - \gamma$), this gives the \emph{density equality}
$f^*(\sigma_i \mid A_{\sigma_i}) / p^*_{\sigma_i}
= f(s \mid A_s) / p_s$, where $A_s$ is the greedy prefix up to $s$ and
$A_{\sigma_i}$ is its image under the split.
Moreover, the equality is stable under iteration: if $s$ itself was
produced by a prior $\gamma'$-split of an ancestor $\tilde s$, then
$f^*(s \mid A_s)/p^*_s = f^{**}(\tilde s \mid A_{\tilde s})/p^{**}_{\tilde s}$
by the same identity applied to that earlier split, and inductively
every fragment $u$ in the descendants of $\tilde s$ satisfies
$f^*(u \mid A_u)/p^*_u = f(\tilde s \mid A_{\tilde s})/p_{\tilde s}$,
i.e.\ the density class of an element is preserved across any
finite sequence of splits applied to it or to disjoint elements
of the prefix.
\end{observation}

\subsection{The \texorpdfstring{$\theta$}{theta}-construction}

In this subsection $U$ and $V$ live in the current split ground set.
We write $p^*_x$ for the cost of a split-ground-set element~$x$; for an
unsplit original item, $p^*_x=p_x$.
Given finite sets $U, V$ with $c(U) \le c(V)$ and a greedy ordering on $V$,
the following iterative procedure builds a mapping $\theta\colon U \to V$
(after possibly splitting some elements of $U$) maintaining the invariant
$\sum_{u \in \theta^{-1}(a)} p^*_u \le p^*_a$ for every $a \in V$.

\begin{definition}[Iterative theta-construction]
\label{def:theta-appendix}
This is the iterative construction of
Feldman and Kuhnle~\cite[Def.~7.15]{FeldmanKuhnle2025}.
Initially $\theta$ is undefined on $U$.
Repeat: if $\theta$ is defined on all of $U$, terminate. Otherwise let
$u$ be the last element of $U$ for which $\theta$ is undefined, and $a$
the last element of $V$ for which
$\sum_{u' \in \theta^{-1}(a)} p^*_{u'} < p^*_a$. If $p^*_u \le p^*_a -
\sum_{u' \in \theta^{-1}(a)} p^*_{u'}$, set $\theta(u) = a$. Otherwise
perform a $\gamma$-split of $u$ with
$\gamma = (p^*_a - \sum_{u' \in \theta^{-1}(a)} p^*_{u'})/p^*_u$, set
$\theta(\sigma_2) = a$, and leave $\sigma_1$ undefined to be processed
in a later iteration.
\end{definition}

\emph{Termination.} Call $a \in V$ \emph{tight} once
$\sum_{u' \in \theta^{-1}(a)} p^*_{u'} = p^*_a$, and let
$T = \#\{a \in V : a \text{ tight}\}$, $D$ the number of $u \in U$
on which $\theta$ is undefined. Each iteration falls into one of
three cases: (i) $\theta(u) = a$ with strict inequality after, in
which case $D$ decreases by $1$ and $T$ is unchanged;
(ii) $\theta(u) = a$ with equality after, in which case $D$ decreases
by $1$ and $T$ increases by $1$; or (iii) $u$ is split into
$\sigma_1, \sigma_2$, $\theta(\sigma_2) = a$ makes $a$ tight, and
$\sigma_1$ replaces $u$ in the undefined set, so $D$ is unchanged
and $T$ increases by $1$. Thus the lexicographic monovariant
$(|V| - T,\, D)$ strictly decreases each iteration and is bounded
below by $(0, 0)$, so the construction terminates after at most
$|V| + |U|$ iterations, of which at most $|V|$ are splits. The
hypothesis $c(U) \le c(V)$ guarantees that whenever some $u$ is
undefined, $\sum_{a \in V} p^*_a > \sum_{a \in V}
\sum_{u' \in \theta^{-1}(a)} p^*_{u'}$, so a non-tight $a$ exists and
the construction never blocks. After all splits, $c(U)$ summed over
the post-split elements still equals the pre-split $c(U)$.

\begin{lemma}[Density-domination lemma]
\label{lem:density-domination}
This lemma is adapted from the density-preservation observation of
Feldman and Kuhnle~\cite[Obs.~7.16]{FeldmanKuhnle2025}.
Let $\theta\colon U \to V$ be the mapping of
Definition~\ref{def:theta-appendix}, where $V$ is a density-greedy-ordered
prefix of a greedy run, possibly with its last item replaced by an included
boundary fragment. Assume every original element from which an element of
$U$ descends was a density-greedy candidate at every step of that run.
Then for every $a \in V$ and every $u \in \theta^{-1}(a)$,
\[
  \frac{f^*(u \mid A_a)}{p^*_u}
  \;\le\;
  \frac{f^*(a \mid A_a)}{p^*_a}.
\]
Consequently,
\[
  \sum_{u \in \theta^{-1}(a)} f^*(u \mid A_a)
  \;\le\;
  \frac{f^*(a \mid A_a)}{p^*_a}
  \sum_{u \in \theta^{-1}(a)} p^*_u
  \;\le\; f^*(a \mid A_a),
\]
where $A_a$ is the greedy prefix strictly preceding $a$ and
the last inequality uses the invariant
$\sum_{u \in \theta^{-1}(a)} p^*_u \le p^*_a$.
\end{lemma}

\begin{proof}[Proof idea]
Let $s(u)$ be the original item from which $u$ descends, and let $b(a)$ be
the original greedy item from which $a$ descends (so $b(a)=a$ if $a$ is not
a boundary fragment). Observation~\ref{obs:density-preservation} gives
\[
  \frac{f^*(u\mid A_a)}{p^*_u}
  =
  \frac{f(s(u)\mid A_a)}{p_{s(u)}},
  \qquad
  \frac{f^*(a\mid A_a)}{p^*_a}
  =
  \frac{f(b(a)\mid A_a)}{p_{b(a)}}.
\]
The original item $s(u)$ was a density-greedy candidate when $b(a)$ was
selected from prefix~$A_a$, hence its density was no larger than the density
of~$b(a)$. This gives the individual density bound. Multiplying it by
$p^*_u$ and summing over $u \in \theta^{-1}(a)$ gives the displayed
consequence; no sum of densities is compared directly to a single density.
\end{proof}

\subsection{Full proof of Theorem~\ref{thm:knapsack-half-containment}}
\label{app:knapsack-proof}

\begin{theorem*}[Theorem~\ref{thm:knapsack-half-containment}, restated]
For any $\varepsilon \in (0, 1/2)$, sequential disjoint density-greedy with
$\ell = \lceil 1/\varepsilon \rceil$ runs produces a set $P \subseteq N$
with $c(P) \le 3 \ell B = O(B/\varepsilon)$ such that for every
$B' \in (0, B]$ there exists $Q \subseteq P$ with
\[
  c(Q) \;\le\; B', \qquad
  f(Q) \;\ge\; \left(\tfrac{1}{2} - \varepsilon\right) \OPT_{B'}
  \;-\; \Delta_{B'},
\]
where $\Delta_{B'} \ge 0$ is a boundary-fragment loss that arises when the
density-greedy prefix crosses the budget threshold mid-item.  When no item
is split, $\Delta_{B'} = 0$.
Total query complexity: $O(n^2/\varepsilon)$.
\end{theorem*}

\paragraph{Proof overview.}
This is the most technical proof in the paper, and the reader may wish to
first read the simpler cardinality proof
(Theorem~\ref{thm:half-containment}, Appendix~\ref{app:half-containment-proof})
for the core ideas.

The main complication beyond cardinality is that density-greedy prefixes may
not align cleanly with a target budget: the last item added may overshoot,
requiring the item to be ``split'' into a used and unused fragment.
The item-splitting machinery of Feldman \&
Kuhnle~\cite{FeldmanKuhnle2025} (stated above) handles this by creating
a new ground set where the split item is replaced by two fragments whose
costs sum to the original.
The proof then proceeds in five steps:

\smallskip
\noindent\emph{Step~1: Boundary split (imported machinery).}
For each density-greedy run, identify the prefix whose total cost equals the
residual budget~$\beta$, splitting the boundary item if needed via
Definition~\ref{def:split-appendix}.
Split-submodularity (Lemma~\ref{lem:split-submod-appendix}) ensures the
resulting function $f^*$ is still submodular.

\smallskip
\noindent\emph{Step~2: Averaging (original, analogous to the cardinality case).}
The disjoint-collection argument from Fact~\ref{fact:submod-disjoint}, now
applied to the split function~$f^*$, identifies the best run $i^\star$ with
$f^*(O^\uparrow \cup A) \ge (1-1/\ell)f(O)$.

\smallskip
\noindent\emph{Step~3: Density domination (imported+original).}
Apply the $\theta$-construction (Definition~\ref{def:theta-appendix}) to map
$\hatO$ into the greedy prefix~$A$.
The density-domination lemma (Lemma~\ref{lem:density-domination}, adapted
from~\cite{FeldmanKuhnle2025}) then yields
$f^*(\hatO^\uparrow \cup A) \le 2f^*(A) - f(\emptyset)$.

\smallskip
\noindent\emph{Step~4: Submodular decomposition (original).}
Apply the lattice inequality
$f^*(X \cup Y) + f^*(X \cap Y) \le f^*(X) + f^*(Y)$ to relate
$f^*(O^\uparrow \cup A)$ to the two candidates.

\smallskip
\noindent\emph{Step~5: Feasible real subset (original).}
Convert the split-function guarantee back to a real-set guarantee, paying the
boundary-fragment loss~$\Delta_{B'}$.

\begin{proof}[Proof of Theorem~\ref{thm:knapsack-half-containment}]
Fix $\varepsilon \in (0, 1/2)$, $\ell = \lceil 1/\varepsilon \rceil$,
and $B' \in (0, B]$. Let $O$ achieve $\OPT_{B'}$; set
$\hatO = O \setminus P$, $\dotO = O \cap P$,
$\beta = c(\hatO) \le B'$.
If $\beta=0$, then $O\subseteq P$, so taking $Q=O$ proves the theorem.
Assume below that $\beta>0$.

\medskip\noindent\textbf{Step 1 (Budget-specific prefix and boundary split).}
For each run $i \in [\ell]$, order $A_i = \{a^{(i)}_1, \ldots,
a^{(i)}_{m_i}\}$ by density-greedy selection time. Let $j_i$ be the first
index for which
\[
  \sum_{r\le j_i} p_{a^{(i)}_r} \;\ge\; \beta .
\]
Since $c(A_i) \ge 2B \ge 2\beta$, such $j_i$ exists and
$\sum_{r < j_i} p_{a^{(i)}_r} < \beta$.
Let
\[
  r_i = \beta-\sum_{r<j_i}p_{a^{(i)}_r},
  \qquad
  \gamma_i = r_i/p_{a^{(i)}_{j_i}} \in (0,1].
\]
If $\gamma_i<1$, apply a $\gamma_i$-split to the boundary item
$a^{(i)}_{j_i}$, producing an included fragment
$\tau_i=\sigma^{(i)}_1$ with cost $p^*_{\tau_i}=r_i$ and an excluded
fragment $\sigma^{(i)}_2$. Define
$A'_i=\{a^{(i)}_1,\ldots,a^{(i)}_{j_i-1},\tau_i\}$.
If $\gamma_i=1$, no split is needed; set
$A'_i=\{a^{(i)}_1,\ldots,a^{(i)}_{j_i}\}$ and say that run~$i$ has no
boundary fragment. In either case $c(A'_i)=\beta$ exactly.
Let $f^*$ denote the function after all actual boundary splits. For an
original set $S\subseteq N$, write $S^\uparrow$ for its full lift in the
split ground set, replacing any split item by both fragments; then
$f^*(S^\uparrow)=f(S)$.
The $A'_i$ are disjoint because the greedy runs use disjoint ground sets.

\medskip\noindent\textbf{Step 2 (Averaging).}
By Fact~\ref{fact:submod-disjoint} applied to the disjoint collection
$\{A'_i\}$ under the split function $f^*$:
\begin{align*}
  \sum_{i=1}^\ell \bigl[f(O) - f^*(O^\uparrow \cup A'_i)\bigr]
  &=
  -\sum_{i=1}^\ell f^*(A'_i \mid O^\uparrow) \\
  &\le
  -f^*\bigl(\textstyle\bigcup_i A'_i \,\big|\, O^\uparrow\bigr) \\
  &=
  f(O) - f^*\bigl(O^\uparrow \cup \textstyle\bigcup_i A'_i\bigr)
  \;\le\; f(O),
\end{align*}
using $f^*(O^\uparrow) = f(O)$ and non-negativity of $f^*$. Averaging
yields some $i^\star$ with
$f^*(O^\uparrow \cup A'_{i^\star}) \ge (1 - 1/\ell) f(O)$.
Write $A = A'_{i^\star}$ and $j^\star=j_{i^\star}$. If run~$i^\star$ has an
included boundary fragment $\tau_{i^\star}$, then
$\tau_{i^\star}$ is a fragment of the original item
$a_{j^\star}^{(i^\star)}$, and its included fraction is
$\gamma_{i^\star}$.  Set
$\widehat A = A \setminus \{\tau_{i^\star}\}
= \{a^{(i^\star)}_1, \ldots, a^{(i^\star)}_{j^\star - 1}\}$ and
$\Delta_{B'}=\gamma_{i^\star} f(\{a_{j^\star}^{(i^\star)}\})$; otherwise set
$\widehat A=A$ and $\Delta_{B'}=0$.

\medskip\noindent\textbf{Step 3 (Density domination).}
Apply Definition~\ref{def:theta-appendix} with source $U = \hatO$ and
target $V = A$ in the current split ground set.
The hypothesis $c(U) = \beta = c(V)$ is met, and every $u \in \hatO$
was a density-greedy candidate at every step of run $i^\star$
(because $\hatO \cap P = \emptyset$ and run $i^\star$'s candidate pool is
$N \setminus \bigcup_{j < i^\star} G_j \supseteq N \setminus P$).
The $\theta$-construction may split elements of~$\hatO$; after it terminates,
let $\hatO^\uparrow$ denote the resulting set of source fragments. Continue
to write $\dotO^\uparrow$ and $O^\uparrow$ for the full lifts after these
additional source splits.
By Lemma~\ref{lem:density-domination}, for each $a \in A$,
$\sum_{u \in \theta^{-1}(a)} f^*(u \mid A_a) \le f^*(a \mid A_a)$, where
$A_a$ is the $f^*$-greedy prefix strictly before $a$ in run $i^\star$.
Summing and applying submodularity of $f^*$
($A_a \subseteq A$ for every $a \in A$),
\begin{align*}
  f^*(\hatO^\uparrow \cup A) - f^*(A)
  &\;=\;
  \sum_{u \in \hatO^\uparrow}
  f^*(u \mid A \cup \{\text{prior } u'\})\\
  &\;\le\;
  \sum_{u \in \hatO^\uparrow} f^*(u \mid A)
  \;\le\;
  \sum_{u \in \hatO^\uparrow} f^*(u \mid A_{\theta(u)}) \\
  &\;\le\; \sum_{a \in A} f^*(a \mid A_a)
  \;=\; f^*(A) - f(\emptyset),
\end{align*}
where the first equality telescopes over any ordering of $\hatO^\uparrow$,
the first two inequalities are submodularity of $f^*$, the third
inequality is Lemma~\ref{lem:density-domination}, and the last equality
telescopes the greedy chain. Conclusion:
$f^*(\hatO^\uparrow \cup A) \le 2 f^*(A) - f(\emptyset)$.

\medskip\noindent\textbf{Step 4 (Submodular decomposition).}
Apply submodularity of $f^*$ to
$X = \hatO^\uparrow \cup A$ and $Y = \dotO^\uparrow \cup A$.
Then $X \cap Y = A$ and $X \cup Y = O^\uparrow \cup A$, so
\[
  f^*(O^\uparrow \cup A) \;+\; f^*(A)
  \;\le\; f^*(\hatO^\uparrow \cup A) \;+\; f^*(\dotO^\uparrow \cup A)
  \;\le\; 2 f^*(A) - f(\emptyset) + f^*(\dotO^\uparrow \cup A).
\]
Combining with Step~2,
\[
  (1 - 1/\ell) f(O) \;-\; f^*(A)
  \;\le\; f^*(O^\uparrow \cup A) - f^*(A)
  \;\le\; f^*(A) - f(\emptyset) + f^*(\dotO^\uparrow \cup A),
\]
so $\max\bigl(f^*(A),\, f^*(\dotO^\uparrow \cup A)\bigr)
\ge \tfrac{1}{2}\bigl[(1 - 1/\ell) f(O) + f(\emptyset)\bigr]
\ge (1/2 - \varepsilon) \OPT_{B'}$.

\medskip\noindent\textbf{Step 5 (Reporting a feasible real subset).}
If no split occurred in run~$i^\star$, then $A=\widehat A$ is already a
real set and $\Delta_{B'}=0$.
When a split did occur, the sets $A$ and $\dotO^\uparrow \cup A$ may contain
the included boundary fragment~$\tau_{i^\star}$; let
$\widehat A^{+}=\widehat A\cup\{a_{j^\star}^{(i^\star)}\}$.
Definition~\ref{def:split-appendix} gives
\begin{equation}
  f^*(A) \;\le\; f(\widehat A) + \Delta_{B'},
  \qquad
  f^*(\dotO^\uparrow \cup A)
  \;\le\; f(\dotO \cup \widehat A) + \Delta_{B'}. \label{eq:split-loss}
\end{equation}
Indeed, in the split case this follows from
$f^*(A)=(1-\gamma_{i^\star})f(\widehat A)
+\gamma_{i^\star} f(\widehat A^+)$ and
$f(\widehat A^+) \le f(\widehat A)+f(\{a_{j^\star}^{(i^\star)}\})$.
For the second inequality, if $a_{j^\star}^{(i^\star)}\in\dotO$ then
$f^*(\dotO^\uparrow\cup A)=f(\dotO\cup\widehat A)$ already; otherwise
$f^*(\dotO^\uparrow\cup A)
=(1-\gamma_{i^\star})f(\dotO\cup\widehat A)
+\gamma_{i^\star} f(\dotO\cup\widehat A^+)$ and submodularity gives
$f(\dotO\cup\widehat A^+) \le
f(\dotO\cup\widehat A)+f(\{a_{j^\star}^{(i^\star)}\})$.
In the no-split case~\eqref{eq:split-loss} is equality.

Feasibility: $c(\widehat A) \le \beta \le B'$ and
$c(\dotO \cup \widehat A) \le c(\dotO) + \beta \le c(O) \le B'$
(both strict in general). Thus $\widehat A$ and $\dotO \cup \widehat A$
are both feasible at budget $B'$. The reported set is
\[
  Q \;=\; \arg\max\!\bigl\{\, f(T) \,:\, T \in \{\widehat A,\;
    \dotO \cup \widehat A\},\; c(T) \le B' \,\bigr\},
\]
giving
\[
  f(Q) \;\ge\; (1/2 - \varepsilon)\OPT_{B'} - \Delta_{B'}.
\]
The boundary item $a_{j^\star}^{(i^\star)}$ belongs to the chosen
density-greedy run for this budget-specific analysis, but it may have
$p_{a_{j^\star}^{(i^\star)}}>B'$ and therefore need not be a feasible
singleton for the query budget. This is why the theorem keeps the loss as
the explicit boundary-fragment term $\Delta_{B'}$ rather than replacing it
by the best feasible singleton.
\end{proof}

\begin{corollary*}[Corollary~\ref{cor:knapsack-small-items}, restated]
Fix a budget range $\mathcal{B}\subseteq(0,B]$ and let
$B_{\min}=\inf \mathcal{B}>0$.
Run sequential disjoint density-greedy with a sufficiently small internal
accuracy parameter.  Suppose that the item sizes are small at the density
prefix scale: for every $B'\in\mathcal{B}$, the boundary item satisfies
$p_{a_{j^\star}^{(i^\star)}} \le \eta \beta$
where $\beta=c(O\setminus P)$.
There is a universal constant $C$ such that, if $\eta\le \varepsilon/C$, then
for every $B'\in\mathcal{B}$ there exists $Q\subseteq P$ with
$c(Q)\le B'$ and
$f(Q)\ge (\tfrac12-\varepsilon)\OPT_{B'}$.
\end{corollary*}

\paragraph{Proof overview.}
The corollary absorbs the boundary-fragment loss $\Delta_{B'}$ from
Theorem~\ref{thm:knapsack-half-containment} under the small-item assumption.
The argument proceeds in two steps, both original:

\smallskip
\noindent\emph{Step~1.}
Bound the true split loss $L_{B'}$ by the density-prefix
average times the fragment cost, using the fact that density-greedy
marginals are nonincreasing.

\smallskip
\noindent\emph{Step~2.}
The small-item assumption gives $r_{i^\star} \le \eta\beta$,
so $L_{B'} \le \eta M$, and choosing $\eta = \varepsilon/8$ absorbs the
loss into the $\varepsilon$ slack.

\begin{proof}[Proof of Corollary~\ref{cor:knapsack-small-items}]
Run the algorithm and the proof above with internal accuracy
$\varepsilon_0=\varepsilon/4$.
Fix $B'\in\mathcal{B}$ and use the notation from the proof of
Theorem~\ref{thm:knapsack-half-containment}.
If no boundary split occurs in the selected run $i^\star$, then
$\Delta_{B'}=0$ and the theorem already gives the desired guarantee after the
constant adjustment.  Suppose a split occurs.

Let
\[
  r_{i^\star}
  = \gamma_{i^\star}p_{a_{j^\star}^{(i^\star)}}
\]
be the cost of the included boundary fragment
$\tau_{i^\star}$, and let
\[
  d_{i^\star}
  =
  \frac{f(a_{j^\star}^{(i^\star)}\mid \widehat A)}
       {p_{a_{j^\star}^{(i^\star)}}}
\]
be the density of its original item at the moment it is selected.
If this marginal is negative, deleting the boundary fragment can only
increase the reported real values, so assume $d_{i^\star}\ge0$.
The true loss from deleting the split fragment is not the crude singleton
bound $\Delta_{B'}$; by the split definition it is
\[
  L_{B'} :=
  \gamma_{i^\star} f(a_{j^\star}^{(i^\star)}\mid \widehat A)
  = r_{i^\star} d_{i^\star}
\]
for the candidate $A$, and at most the same quantity for
$\dotO^\uparrow\cup A$ by submodularity:
\[
  f^*(\dotO^\uparrow\cup A)-f(\dotO\cup\widehat A)
  \le L_{B'} .
\]

We now bound this loss by the density-prefix average.  Along the
density-greedy order defining $A$, marginal densities are nonincreasing; dummy
zero-marginal items ensure that the prefix of cost at most $B$ used here has
nonnegative densities unless the boundary loss is already nonpositive.
Therefore the boundary density is at most the average density of the split
prefix:
\[
  d_{i^\star}
  \le
  \frac{f^*(A)-f(\emptyset)}{\beta}.
\]
The small-item assumption gives
$r_{i^\star}\le p_{a_{j^\star}^{(i^\star)}}\le \eta\beta$, and hence
\[
  L_{B'} \le \eta\bigl(f^*(A)-f(\emptyset)\bigr)
  \le \eta f^*(A),
\]
using nonnegativity of $f(\emptyset)$.

Let
\[
  M=\max\{f^*(A),\,f^*(\dotO^\uparrow\cup A)\}.
\]
Step~4 of the theorem proof gives
$M\ge(1/2-\varepsilon_0)\OPT_{B'}$, while the preceding paragraph shows that
replacing the split candidate(s) by the feasible real candidates
$\widehat A$ and $\dotO\cup\widehat A$ loses at most
$L_{B'}\le\eta f^*(A)\le\eta M$.  Thus one of these real feasible candidates
has value at least
\[
  (1-\eta)M
  \ge
  (1-\eta)(1/2-\varepsilon_0)\OPT_{B'}.
\]
Choosing $\eta\le\varepsilon/8$ and $\varepsilon_0=\varepsilon/4$ yields
$(1/2-\varepsilon)\OPT_{B'}$, after adjusting constants.
\end{proof}

\medskip\noindent\textbf{Remark on the boundary term.}
The additive loss
$\Delta_{B'}=\gamma_{i^\star} f(\{a_{j^\star}^{(i^\star)}\})$ is inherent
to this density-greedy knapsack analysis with item splitting: when the
residual budget crosses a large boundary item, the split fraction
$\gamma_{i^\star}$ can approach~$1$ and the loss approaches
$f(\{a_{j^\star}^{(i^\star)}\})$.
In the \emph{algorithmic} (non-pruning) setting, standard
fixes---guessing the top constant-many items by enumeration
(see~\cite{FeldmanKuhnle2025}, Section~7)---preserve the clean
$(1/2 - \varepsilon)$ ratio, but enumeration is incompatible with pruning
because the guess depends on the query budget~$B'$.
A fully bicriteria variant (allowing $c(Q) \le (1 + \varepsilon) B'$)
admits $\widehat A^{+}$ directly and removes the boundary term without
assumption.

\section{Warm-Start BFNS Barrier}
\label{app:warm-start}

\paragraph{Overview.}
This section explores whether combining the averaging idea of
Theorem~\ref{thm:half-containment} (Step~1) with the CDG algorithm
of~\cite{BFNS2014} (Theorem~\ref{thm:bfns-cdg} in Appendix~\ref{app:deferred-prelim})
can beat the $1/2$ barrier.
The CDG guarantee (imported) supplies a three-term bound that converts a
warm-start value $g(\emptyset) = f(A)$ into additional containment;
the analysis below (original) shows that the adversary's value-split
degree of freedom limits the combined guarantee to
$\tfrac{4}{9}(1-1/\ell)$, which exceeds $0.401$ only for small warm-start
sizes $b/k \le 0.35$.  Beyond that breakpoint, the direct
$0.401$-approximation of~\cite{BuchbinderFeldman2024} dominates.

\medskip
For each $b \in \{0, 1, \ldots, k\}$, run $\ell(b)$ disjoint greedy rounds
of size $b$ on $N$, producing sets $A_1, \ldots, A_{\ell}$.
By the averaging argument (Step~1 of Theorem~\ref{thm:half-containment}),
the best $A_{\hat\imath}$ satisfies
$f(O \cup A_{\hat\imath}) \ge (1 - 1/\ell)\, \OPT$.

Apply CDG (Theorem~\ref{thm:bfns-cdg}) to
$g(S) = f(S \cup A_{\hat\imath})$ over
$G \setminus A_{\hat\imath}$ with budget $k - b$,
at density tuned to $v = 4$ (i.e., $|G \setminus A_{\hat\imath}| = 5(k-b)$).
The warm-start gives $g(\emptyset) = f(A_{\hat\imath}) = V$.

At $v = 4$, Theorem~\ref{thm:bfns-cdg} gives equal coefficients:
\[
  g(D) \;\ge\; \tfrac{4}{9}\, g(\OPT_g) + \tfrac{4}{9}\, g(\emptyset)
  = \tfrac{4}{9}\bigl(W + V\bigr),
\]
where $W = g(\OPT_g) = \max_{|T| \le k-b} f(T \cup A_{\hat\imath})$
and $V = f(A_{\hat\imath})$.
From the submodularity decomposition:
$(1 - 1/\ell)\,\OPT \le f(O \cup A_{\hat\imath}) \le V + W$.
Therefore:
\[
  f(\text{output}) \;\ge\; \tfrac{4}{9}\,(1 - 1/\ell)\,\OPT.
\]
The guarantee is independent of the value split between $V$ and $W$.

The number of greedy rounds needed to achieve $v = 4$ is
$\ell(b) = 5(k-b)/b + 1$.
The overall guarantee is $\tfrac{4}{9}(1 - 1/\ell)$:

\begin{center}
\begin{tabular}{@{}ccccc@{}}
\toprule
$b/k$ & $\ell$ & Total elements & Guarantee & $> 0.401$? \\
\midrule
$0.10$ & $46$ & $\sim 460k$ & $0.434$ & Yes \\
$0.20$ & $21$ & $\sim 420k$ & $0.423$ & Yes \\
$0.30$ & $12.7$ & $\sim 380k$ & $0.409$ & Yes \\
$0.35$ & $10.3$ & $\sim 360k$ & $0.401$ & Breakeven \\
$0.40$ & $8.5$ & $\sim 340k$ & $0.392$ & No \\
\bottomrule
\end{tabular}
\end{center}

When $|\hatO|$ is large (many OPT elements outside $G$), the budget $b$ must be
large to dominate them, forcing $\ell$ small, and the $(1-1/\ell)$ factor
eats the $4/9$ gain.
The crossover is at $b/k \approx 0.35$; for larger $b$, one falls back to
the direct $0.401$ of~\cite{BuchbinderFeldman2024}.
Adding a random-greedy fallback branch does not help: the adversary controls
the value split $\beta = f(\dotO)/\OPT$, and for any density $c$,
the guarantee is at most $1/e$ in the worst case.

\begin{observation}
For any fixed warm-start size $b$, the CDG inner density $v = 4$ is the
\emph{unique} optimal choice of the inner ground-set size parameter.
Writing $\alpha_v, \beta_v$ for the CDG coefficients on $g(\OPT_g)$ and
$g(\emptyset)$ respectively (Theorem~\ref{thm:bfns-cdg}):
at $v < 4$, $\alpha_v > \beta_v$ and the adversary sets $V = 0$;
at $v > 4$, $\beta_v > \alpha_v$ and the adversary sets $W = 0$.
Only at $v = 4$ are the coefficients equal, eliminating the adversary's
value-split degree of freedom.
The $b/k$ table above holds $v = 4$ fixed and varies the warm-start size;
the bound $\tfrac{4}{9}(1-1/\ell(b))$ then peaks at small $b/k$, where
$\ell$ is large.
\end{observation}

\section{Window Containment Analysis}
\label{app:window-analysis}

\begin{algorithm}[H]
\caption{Window Containment}
\label{alg:window-greedy}
\KwIn{Ground set $N$, oracle for $f$, constraint $k$, window parameter $\omega \ge 1$}
\KwOut{Pruned set $P_k$ with $|P_k| \le k+\omega k^2$}
$S_0 \gets \emptyset$\;
$U_0 \gets \emptyset$\;
\For{$t = 0, 1, \ldots, k-1$}{
  $W_t \gets$ top-$\min(\omega k,\, |N \setminus S_t|)$ elements of $N \setminus S_t$ by $f(\cdot \mid S_t)$\;
  $U_{t+1} \gets U_t \cup W_t$\;
  $w_t \gets$ uniformly random element of $W_t$\;
  $S_{t+1} \gets S_t \cup \{w_t\}$\;
}
\Return{$P_k = U_k \cup S_k = S_k \cup \bigcup_{t=0}^{k-1} W_t$}\;
\end{algorithm}

\begin{table}[h]
\centering
\begin{tabular}{@{}cccc@{}}
\toprule
$\omega$ & Window size $\omega k$ & $e^{-1/\omega}/2$ & Output $|P_k| \le k+\omega k^2$ \\
\midrule
$1$ & $k$ & $0.184$ & $k+k^2$ \\
$2$ & $2k$ & $0.303$ & $k+2k^2$ \\
$5$ & $5k$ & $0.410$ & $k+5k^2$ \\
$10$ & $10k$ & $\mathbf{0.452}$ & $k+10k^2$ \\
$20$ & $20k$ & $0.476$ & $k+20k^2$ \\
$100$ & $100k$ & $0.495$ & $k+100k^2$ \\
\bottomrule
\end{tabular}
\caption{Window containment guarantee vs.\ window parameter $\omega$ (non-monotone, unconditional).
At $\omega = 5$ the guarantee exceeds BF's $0.401$; at $\omega = 20$
it approaches Qi's $0.478$ hardness.}
\label{tab:window}
\end{table}

\paragraph{Proof overview.}
The analysis is entirely original and proceeds through two supporting
lemmas and a main proof.
The key tension: to dominate missed OPT elements, we want to pick the
best marginal-gain element at each step; but to avoid destroying OPT's
value, we want to pick randomly (so no single OPT element is too likely
to be displaced).
The window resolves this by ensuring \emph{every} element in the top-$\omega k$
window dominates every element outside it---so a random pick from the window
simultaneously achieves both goals.

\smallskip
\noindent\emph{Lemma~1: Window domination.}
At each step, the window $W_t$ contains the $\omega k$ elements with highest
marginal gain.
Any OPT element outside the window is therefore dominated by every window
element, including the random pick $w_t$.
A telescoping argument (analogous to Step~2 of the SDG proof) yields
$2f(S_m) - f(\emptyset) \ge f(D \cup S_m)$, where $D$ collects the dominated elements.

\smallskip
\noindent\emph{Lemma~2: Degradation bound.}
Picking uniformly at random from a window of size $\omega k$ reduces the
residual OPT value by a factor of $(1 - 1/(\omega k))$ per step.
After $k$ steps, this compounds to $e^{-1/\omega}$ by
Fact~\ref{fact:one-minus-compound}.

\smallskip
\noindent\emph{Main proof.}
Combine the two lemmas via the lattice inequality, just as in the SDG proof.
Both candidate sets ($S_m$ and $\dotO \cup S_m$) are feasible subsets of the
accumulated window~$P_k$.

\begin{fact}[Compound decay]
\label{fact:one-minus-compound}
For $x \ge 2$ and integer $0\le m\le x$:
\[
  (1 - 1/x)^m \ge e^{-m/x}\bigl(1 - m/x^2\bigr).
\]
In particular, with $x = \omega k$ and $m \le k$:
\[
  (1 - 1/(\omega k))^m
  \ge e^{-1/\omega}(1 - 1/(\omega^2 k)).
\]
This follows from the elementary bound
$\ln(1-u)\ge -u-u^2$ for $0\le u\le 1/2$, applied with $u=1/x$,
and then $e^{-z}\ge 1-z$.
\end{fact}

\begin{lemma}[Window Domination]
\label{lem:distinct-pairing}
Fix any set $O$ with $|O| \le k$.
Run the retirement process in the proof below, and let $m \le |O|$
be its stopping time.
There exists an injective pairing $\pi\colon \{0,\ldots,m-1\} \hookrightarrow O$
with $\pi(t) \in \hatO_t$ (the unretired outside-window set from that
process) and $f(w_t \mid S_t) \ge f(\pi(t) \mid S_t)$
for all $t < m$.
Telescoping gives $2\,f(S_m) - f(\emptyset) \ge f(D \cup S_m)$,
where $D = \{\pi(0), \ldots, \pi(m-1)\}$ is the set of dominated OPT elements.
\end{lemma}

\begin{lemma}[Degradation Bound]
\label{lem:degradation}
In the nontrivial case where every window used before the stopping time has
size $\omega k$, uniformly random selection from $W_t$ gives
$\mathbb{E}[f(O \cup S_t)] \ge (1 - 1/(\omega k))\,f(O \cup S_{t-1})$.
Iterating $m \le k$ steps:
\[
  \mathbb{E}[f(O \cup S_m)]
  \ge (1 - 1/(\omega k))^k\, f(O)
  = e^{-1/\omega}\bigl(1 - O(1/(\omega^2 k))\bigr)\,f(O).
\]
\end{lemma}

\begin{proof}[Proof of Lemma~\ref{lem:distinct-pairing}]
We construct the pairing $\pi$ greedily.
Let $R_0 = O$ (all OPT elements are ``unretired'').

At step $t$, define $\hatO_t = R_t \setminus (S_t \cup W_t)$:
the unretired OPT elements outside both the solution and the window.
Since $W_t$ consists of the top-$(\omega k)$ elements by $f(\cdot \mid S_t)$,
every $o \in \hatO_t$ satisfies $f(w \mid S_t) \ge f(o \mid S_t)$
for \emph{all} $w \in W_t$.
In particular, the uniformly random pick $w_t$ dominates $o$.

If $\hatO_t \neq \emptyset$: choose $\pi(t) \in \hatO_t$ arbitrarily
and set $R_{t+1} = R_t \setminus \{\pi(t)\}$.

If $\hatO_t = \emptyset$: every unretired OPT element is in
$W_t \cup S_t$; set $m = t$ and \textbf{stop}.

Since $|R_0| = |O|$ and we remove one element per non-stopping step,
and $\hatO_t \subseteq R_t$, the process terminates at some $m \le |O|$.
At termination, the remaining set $\dotO := R_m = O \setminus D$
(where $D = \{\pi(0), \ldots, \pi(m-1)\}$) satisfies
either $\dotO = \emptyset$ (all OPT elements have been retired) or
$\dotO \subseteq W_m \cup S_m$.

For the value bound, telescope over the $m$ paired steps:
\begin{align*}
  f(S_m) - f(\emptyset)
  &= \sum_{t=0}^{m-1} f(w_t \mid S_t)
  \ge \sum_{t=0}^{m-1} f(\pi(t) \mid S_t)
  \ge \sum_{t=0}^{m-1} f(\pi(t) \mid S_m) \\
  &\ge f\!\left(\bigcup_{t=0}^{m-1} \{\pi(t)\} \;\middle|\; S_m\right)
  = f(D \mid S_m) = f(D \cup S_m) - f(S_m),
\end{align*}
where the second inequality uses submodularity ($S_t \subseteq S_m$)
and the third uses submodularity on individual vs.\ joint marginals.
Rearranging: $2\,f(S_m) - f(\emptyset) \ge f(D \cup S_m)$.
\end{proof}

\begin{proof}[Proof of Lemma~\ref{lem:degradation}]
If some window contains all remaining elements, then
$P_k=N$ from that step onward and containment is exact.  Thus assume
$|W_t|=\omega k$ for the steps considered below.
By submodularity applied to the $\omega k$ window elements at step $t$:
\[
  \sum_{v \in W_t} \bigl[f(O \cup S_t \cup \{v\}) - f(O \cup S_t)\bigr]
  \ge f(O \cup S_t \cup W_t) - f(O \cup S_t)
  \ge -f(O \cup S_t),
\]
where the last step uses non-negativity of $f$.
For a uniformly random $w_t \in W_t$ (size $\omega k$):
\[
  \mathbb{E}[f(O \cup S_{t+1}) \mid S_t]
  \ge (1 - 1/(\omega k))\, f(O \cup S_t).
\]
Iterating over $m \le k$ steps and applying
Fact~\ref{fact:one-minus-compound} with $x = \omega k$ and $m \le k$:
\[
  \mathbb{E}[f(O \cup S_m)]
  \ge (1 - 1/(\omega k))^k\, f(O)
  \ge e^{-1/\omega}(1 - 1/(\omega^2 k))\, f(O).
\]
\end{proof}

\begin{theorem*}[Theorem~\ref{thm:window-containment}, restated]
For any non-negative submodular $f$ and window parameter $\omega \ge 1$,
Algorithm~\ref{alg:window-greedy} returns an accumulated-window set
$P_k \subseteq N$ with $|P_k| \le k+\omega k^2$ such that, for every
$k' \le k$,
\[
  \mathbb{E}\!\left[
    \max_{\substack{|T| \le k'\\ T \subseteq P_k}} f(T)
  \right]
  \;\ge\; \frac{e^{-1/\omega}}{2} \cdot \OPT_{k'}
  \qquad\text{(asymptotically as $k \to \infty$)}.
\]
Total query complexity: $O(nk)$.
\end{theorem*}

\begin{proof}[Proof of Theorem~\ref{thm:window-containment}]
Algorithm~\ref{alg:window-greedy} returns
$P_k = S_k \cup \bigcup_{t=0}^{k-1} W_t$, so
$|P_k| \le k+\omega k^2$.
Fix any target budget $k' \le k$ and let $O$ be an optimal feasible solution
with $|O| \le k'$.
At the termination step $m \le k'$ of Lemma~\ref{lem:distinct-pairing},
the dominated set $D$ and the remaining set $\dotO = O \setminus D$ satisfy
$D \cap \dotO = \emptyset$.
The lemma's termination condition gives either $\dotO=\emptyset$ or
$\dotO \subseteq W_m \cup S_m$; in both cases $\dotO \subseteq P_k$,
because every window $W_m$ used by the process is accumulated in~$P_k$
and $S_m \subseteq S_k \subseteq P_k$.

Both $S_m$ and $\dotO \cup S_m$ are $k'$-feasible subsets of $P_k$:
$|S_m| = m \le k'$ and
$|\dotO \cup S_m| \le |\dotO|+m \le |O| \le k'$.

By submodularity with $A = \dotO \cup S_m$ and $B = D \cup S_m$
(so $A \cap B = S_m$ and $A \cup B = O \cup S_m$):
\[
  f(O \cup S_m) + f(S_m) \le f(\dotO \cup S_m) + f(D \cup S_m).
\]
Substituting the domination bound $f(D \cup S_m) \le 2f(S_m) - f(\emptyset)$:
\[
  f(\dotO \cup S_m) \ge f(O \cup S_m) - f(S_m) + f(\emptyset),
\]
so $\max(f(\dotO \cup S_m),\, f(S_m)) \ge (f(O \cup S_m) + f(\emptyset))/2$.
In expectation, Lemma~\ref{lem:degradation} gives
$\mathbb{E}[f(O \cup S_m)] \ge e^{-1/\omega}(1 - O(1/(\omega^2 k)))\, \OPT_{k'}$,
yielding the claimed bound asymptotically in $k$.
Since $k' \le k$ was arbitrary, the same accumulated set $P_k$
certifies containment at every scale.
\end{proof}

\paragraph{Why random selection resolves the tension.}
The earlier conditional version required max-marginal selection for
domination and random selection for degradation---conflicting requirements.
Window containment resolves this because the top-$w$ construction ensures
\emph{every} window element dominates every non-window element, including
the uniformly random pick. No max-marginal selection within $W_t$ is needed.

\section{Additional Open Problems}
\label{app:open}

The two central open problems (tight non-monotone containment and
extraction beyond the $0.401$ maximization barrier) are stated in
\S\ref{sec:conclusion}.
We collect here three further structural questions suggested by our
proofs.

\begin{openquestion}[Non-monotone lower bounds]
Develop pruning lower-bound techniques beyond density composition.
Composition is structurally inadequate for non-monotone due to the flat
approximation landscape.
Information-theoretic arguments about distinguishing ``useful'' from ``poisonous''
elements may be needed.
\end{openquestion}

\begin{openquestion}[Density-dependent hardness for non-monotone]
Is there a density-dependent inapproximability result for non-monotone
cardinality-constrained maximization?
The current hardness of $0.478$ (Qi) is essentially density-independent.
A density-dependent hardness that approaches $1/e$ at low density would
restore the composition technique's power for non-monotone pruning.
\end{openquestion}

\begin{openquestion}[Monotonicity of the top-$k$ set]
\label{oq:monotone-window}
If $M_i$ (the top $k$ marginals at step $i$ of random greedy) could be
shown to induce a ``locally monotone'' subproblem---that is, if $f$
restricted to subsets of $M_i$ were monotone---then $(1-1/e)$
containment would follow immediately.
OPT ``witnesses'' this: $f$ restricted to subsets of $O$ is monotone
(by optimality).
Can this structural property be extended to algorithmically identified
sets without knowledge of $O$?
\end{openquestion}

\section{Additional Experimental Results}
\label{app:snap-results}

\paragraph{SNAP social networks (MaxCut).}
On Facebook ($n=4039$, $m=88{,}234$) and Wiki-Vote ($n=7115$, $m=100{,}762$)
with $k=50$, SeqDisjoint, WindowMax, WindowRand, and StdGreedy all round to
$\alpha = 1.000$ at every budget $\omega \in \{2, 3, 5\}$
against the full-ground-set greedy reference (all are at least $0.9995$).
QuickPrune plateaus at $\alpha \approx 0.87$; random drops to
$\alpha = 0.13$ on Wiki-Vote at $\omega = 2$.
The structured algorithms consistently recover the greedy reference even
when $|P| \ll n$.
Since exact MaxCut is NP-hard on these graphs, OPT is estimated by the
greedy cut on the full ground set; reported $\alpha$ values are therefore
upper bounds on the true containment ratio.

\paragraph{Scaling to $n = 500$ (LLM context).}
Expanding the HotpotQA ground set to $n = 500$ passages ($k = 5$, 50~questions),
SeqDisjoint and StdGreedy maintain FL containment $\alpha = 1.000$
at every $\omega$ and achieve F1~$= 0.48$ at $\omega = 2$
($50\times$ compression)---$3\times$ higher than random (F1~$= 0.16$).
A stochastic thresholded variant uses $2.2\times$
fewer oracle queries and matches at $\omega \ge 5$.

\paragraph{Comparison with learned pruning (COMBHelper).}
We compare against COMBHelper~\cite{TianEtAl2024}, a recent GNN-based
method that trains a student network via knowledge distillation to predict
which vertices belong to the optimal solution.
Since COMBHelper supports only vertex cover and MIS, we adapt its
architecture (4-layer GraphSAGE teacher, 3-layer student) to MaxCut,
training on 20 Erd\H{o}s--R\'{e}nyi graphs $G(100, 500)$ with exact
$k$-MaxCut labels from an IP solver.
Both networks use GraphSAGE layers with sum aggregation, normalization, ReLU
between layers, and dropout probability~0.5.
Node features are degree, degree divided by the maximum degree, and local
clustering coefficient.
The teacher has hidden width~64 and the student has hidden width~32.
We train with Adam, learning rate~0.01, weight decay~$5\cdot10^{-4}$, and
200~epochs.  Training labels are binary membership indicators for the exact
$k$-MaxCut solution; each graph uses a random 50/50 node train/validation
split.  The student is trained by knowledge distillation from the teacher with
temperature $T=4$ and loss
$0.5\,\mathcal{L}_{\mathrm{CE}} + 0.5\,\mathcal{L}_{\mathrm{KD}}$.
The reported variance check retrains the whole pipeline with training seeds
$42,123,256,512,1024$.
Table~\ref{tab:combhelper} reports containment ratios alongside our
provable algorithms.

\begin{table}[h]
\centering
\small
\begin{tabular}{@{}l c cccc c@{}}
\toprule
Instance & $\omega$ & CH-Teacher & CH-Student & SeqDisjoint & StdGreedy & Random \\
\midrule
$G(100, 1000)$  & 2 & 0.997 & 0.997 & 0.997 & 0.997 & 0.877 \\
$G(200, 2000)$  & 2 & 1.000 & 1.000 & \textbf{1.000} & \textbf{1.000} & 0.825 \\
$G(300, 4500)$  & 2 & 0.999 & 0.999 & \textbf{1.000} & \textbf{1.000} & 0.845 \\
\midrule
SNAP/Facebook   & 2 & 0.975 & 0.998 & \textbf{1.000} & \textbf{1.000} & 0.255 \\
SNAP/Facebook   & 3 & 0.978 & 1.000 & \textbf{1.000} & \textbf{1.000} & 0.379 \\
SNAP/Facebook   & 5 & 0.999 & 1.000 & \textbf{1.000} & \textbf{1.000} & 0.523 \\
\bottomrule
\end{tabular}
\caption{COMBHelper~(learned) vs.\ provable pruning on MaxCut.
On in-distribution random graphs, COMBHelper matches our algorithms.
On the out-of-distribution SNAP/Facebook graph ($n=4039$), the
teacher network drops to $\alpha = 0.975$ at $\omega = 2$ while
SeqDisjoint and StdGreedy maintain $\alpha = 1.000$.}
\label{tab:combhelper}
\end{table}

On random graphs drawn from the same family as the training data,
COMBHelper is competitive---both teacher and student achieve
$\alpha \ge 0.997$.
The gap appears on out-of-distribution inputs: on the SNAP/Facebook
graph, which differs in scale ($n = 4039$ vs.\ training on $n = 100$)
and topology (power-law vs.\ Erd\H{o}s--R\'{e}nyi),
the teacher drops to $\alpha = 0.975$ at $\omega = 2$.
By contrast, our provable algorithms maintain $\alpha = 1.000$ without
any training data or graph-specific adaptation.
Beyond the generalization gap, learning-based pruning faces a structural
limitation: COMBHelper requires \emph{exact solutions} on training
instances to generate supervision labels---precisely the expensive
computation that pruning is meant to avoid.
Our algorithms sidestep this circular dependency entirely,
requiring only query access to~$f$.
We note that COMBHelper's in-distribution performance is robust to
initialization: retraining with 5~seeds yields identical or near-identical
$\alpha$ on the test instances (maximum std $0.0008$), so the results above
are not a lucky seed.

\paragraph{LLM proxy containment on HotpotQA.}
Table~\ref{tab:llm-proxy} reports proxy containment on HotpotQA with expanded ground sets.

\begin{table}[h]
\centering
\small
\begin{tabular}{@{}l cccc c@{}}
\toprule
Algorithm & $\omega=2$ & $\omega=3$ & $\omega=5$ & $\omega=8$
          & $|P|$ at $\omega\!=\!5$ \\
\midrule
SeqDisjoint & \textbf{1.000} & \textbf{1.000} & \textbf{1.000} & \textbf{1.000} & 15 \\
WindowMax   & \textbf{1.000} & \textbf{1.000} & \textbf{1.000} & \textbf{1.000} & 7 \\
WindowRand  & 0.984 & 0.985 & 0.979 & 0.978 & 7 \\
QuickPrune  & 0.836 & 0.836 & 0.836 & 0.836 & 3 \\
StdGreedy   & \textbf{1.000} & \textbf{1.000} & \textbf{1.000} & \textbf{1.000} & 15 \\
Random      & 0.791 & 0.869 & 0.933 & 0.966 & 15 \\
\bottomrule
\end{tabular}
\caption{Proxy containment on HotpotQA with expanded ground sets ($n=50$, $k=3$,
50~questions).
WindowMax achieves $\alpha = 1.0$ with $|P| = 7$ ($1.4\times$ compression
vs.\ $|P|=15$ for SeqDisjoint at $\omega = 5$).
QuickPrune outputs only $|P| \approx 3$ elements regardless of budget---the
streaming threshold is too conservative on non-monotone inputs.}
\label{tab:llm-proxy}
\end{table}

\paragraph{LLM experiment details.}
\label{app:llm-details}

\emph{Reader model and decoding.}
All LLM reader evaluations use \texttt{openai/gpt-oss-120b} served locally
through the OpenAI-compatible vLLM API.  The model is the Apache~2.0
open-weight reader listed in Appendix~\ref{app:assets}.  We query two local
endpoints in round-robin order, \texttt{localhost:8000/v1} and
\texttt{localhost:8003/v1}.  Decoding is deterministic:
temperature~$0.0$ and \texttt{max\_tokens}$=256$.
For every evaluated passage set, the prompt template is:
\begin{quote}
\small
\ttfamily
Answer the following question based on the provided context. Give a short,
direct answer.\par
\medskip
Context:\par
[title] passage text\par
\medskip
Question: question text\par
\medskip
Answer:
\end{quote}
For multiple passages, the context block concatenates titled passages separated
by blank lines.  For \texttt{gpt-oss} responses we read the first nonempty field
among \texttt{message.content}, \texttt{reasoning}, and
\texttt{reasoning\_content}, matching the local vLLM behavior.

\emph{Scoring.}
HotpotQA answers are scored by normalized token F1 and exact match: lowercase,
remove punctuation and the articles ``a'', ``an'', and ``the'', collapse
whitespace, then compute token overlap F1 or exact normalized-string equality.
MuSiQue LLM validation uses the same normalization, but reports the best F1/EM
over the official answer and all answer aliases.

\emph{Submodularity of the proxy.}
The function $\hat{f}(S)=\mathrm{FL}(S)-\theta(|S|)$ is submodular
because the decreasing marginal gains of $\mathrm{FL}$ and the
increasing marginal penalty of convex~$\theta$ reinforce each other:
for any $A \subseteq B$ and element $e \notin B$,
$\hat{f}(A \cup \{e\}) - \hat{f}(A) \ge
 \hat{f}(B \cup \{e\}) - \hat{f}(B)$,
since the $\mathrm{FL}$ marginal shrinks while the $\theta$ marginal grows.

\emph{HotpotQA ground-set construction.}
For each question, we pool its 10~original passages (2~gold, 8~distractors)
with distractor passages from other questions, giving $n = 50$ ($k = 3$)
and $n = 500$ ($k = 5$).
We evaluate on the fitted proxy $\hat{f}(S) = \mathrm{FL}(S) - \theta(|S|)$,
where $\mathrm{FL}$ is facility location on TF-IDF similarities and
$\theta$ is fitted from LLM evaluation data.  The fitted-proxy file used in
the paper, \texttt{fitted\_proxy\_200q.json}, is based on
\texttt{exp2\_results\_200q.json}: 200 HotpotQA development questions, one
random prefix ordering per question, and 1{,}992 recorded LLM evaluations over
context sizes 1--10.  All but one question have all 10 prefix sizes recorded;
one question has sizes 1 and 2 only.  This is why the fitted-proxy file reports
\texttt{n\_results}$=1992$, not $200\times 10$.
The proxy-containment entries in Tables~\ref{tab:llm-proxy}
and~\ref{tab:llm-validation} use \texttt{exp4\_results.json}.  The reader
F1/EM entries for the $\omega=2$ StdGreedy and SeqDisjoint rows use
\texttt{exp5\_llm\_c2\_50q.json}; the $\omega=5$, QuickPrune, window,
random, gold-only, and all-original rows use
\texttt{exp5\_llm\_validation\_50q.json}, which records 451 reader calls
including the initial endpoint sanity check.
The proxy is not the reader score itself; it is a cheap objective used for
pruning, and we separately validate the selected passages with the LLM.
Thus proxy containment only says that the reduced pool preserves high-scoring
candidate contexts under $\hat f$; the reader experiment checks whether those
candidate contexts are useful for answering.

\emph{MuSiQue pilot setup.}
We use the MuSiQue v1.0 answerable development split, keep questions with at
least three decomposition hops, shuffle with seed~42, and take the first
50~questions.  Each question supplies its native 20 candidate paragraphs; gold
passages are the paragraphs marked as supporting.  We evaluate
$k\in\{3,5,7,10\}$ and $\tau\in\{0.1,0.2,0.3,0.5\}$.  Passage embeddings are
computed with \texttt{all-MiniLM-L6-v2}; top-$k$ TF-IDF and top-$k$ embedding
baselines rank passages by query similarity.  The MuSiQue reader validation
uses the selections from the pilot, methods
\texttt{topk\_tfidf}, \texttt{topk\_embedding},
\texttt{rfl\_emb\_tau=0.3}, and \texttt{fl\_qw\_emb}, and
$k\in\{5,10\}$, for 400 deterministic reader calls.

\emph{Restricted facility-location (RFL) function.}
We define
$f_\tau(S) = \sum_{v:\,\mathrm{rel}(v,q) > \tau}
  \max_{s \in S}\mathrm{sim}(v,s)$,
which restricts FL coverage to passages above a relevance threshold~$\tau$,
diversifying \emph{within} the query-relevant neighborhood.
The headline method uses embedding cosine similarity for $\mathrm{sim}$ and
embedding query similarity for $\mathrm{rel}$, with $\tau=0.3$.  We also report
TF-IDF-gated variants in the result file.  Gold-passage recall is the fraction
of supporting passages selected.

\emph{Per-$k$ breakdown and gating comparison.}
At $\tau = 0.3$, RFL with embedding-based gating
trails top-$k$ embedding at $k=3$, ties at $k=5$, and wins at $k=10$,
consistent with the diversity mechanism becoming effective only when
the number of selected passages exceeds the number of trivially salient ones.
At $k=10$, RFL reaches recall $0.740$, which is $2.0$ percentage points
above top-$k$ embedding ($0.720$) and $9.2$ percentage points above
top-$k$ TF-IDF ($0.648$).
TF-IDF-based gating fails---it filters out later-hop gold passages whose
vocabulary diverges from the query---confirming that semantic embedding
similarity is essential for bridging entity gaps in multi-hop reasoning.

\emph{Bootstrap confidence intervals.}
The headline MuSiQue confidence interval is a paired question-level bootstrap
over the 50 pilot questions.  For each of 10{,}000 resamples with replacement
(seed~42), we resample questions and recompute the mean gold-recall difference
between \texttt{rfl\_emb\_tau=0.3} and the baseline.  The reported
$[+2.8,+15.5]$ percentage-point interval is the 2.5/97.5 percentile interval for
\texttt{rfl\_emb\_tau=0.3} minus \texttt{topk\_tfidf} at $k=10$.

\paragraph{Reader saturation.}
\label{app:reader-saturation}
Despite the $+9.2$ percentage-point gold-recall advantage, LLM answer F1 does not
improve: at $k=10$, top-$k$ embedding achieves F1 $= 0.234$ while
RFL achieves F1 $= 0.186$ using the local 120B reader.
This \emph{reader-saturation} effect---the reader model extracts answers
from the most salient passages regardless of additional evidence---is
consistent with the ``lost in the middle'' findings
of~\cite{LiuLinHewitt2024}: providing more gold passages helps recall
but does not help a reader that already attends preferentially to early
or highly salient context.
Improving the reader's ability to synthesize dispersed evidence is an
important open direction; our contribution here is the evidence-coverage
mechanism itself, which the pruning theory certifies.

\paragraph{Synthetic separation: main-body summary table.}
Table~\ref{tab:separation-main} presents the summary statistics for the
synthetic non-monotone separation instances discussed in the main text.

\begin{table}[h]
\centering
\small
\begin{tabular}{@{}cc ccc cc@{}}
\toprule
$n$ & $k$ & $\omega$ & Instances
& $\OPT_k \subseteq P_g$ & $\OPT_k \subseteq P_s$ & Value seps \\
\midrule
20 & 3 & 2 & 2000 & 60\% & 78\% & 123 (6.2\%) \\
30 & 3 & 2 & 1000 & 58\% & 73\% & 86 (8.6\%) \\
50 & 3 & 5 & 500  & 70\% & \textbf{100\%} & 48 (9.6\%) \\
50 & 5 & 3 & 200  & 41\% & 67\% & 33 (16.5\%) \\
\bottomrule
\end{tabular}
\caption{Synthetic non-monotone separation instances (summary).
$P_g$ is the pruned set from one greedy run; $P_s$ is the pruned set from
SeqDisjoint.  \emph{Value seps} counts instances where the solution extracted
from $P_s$ has strictly higher value than the solution extracted from $P_g$.
Full results are in Table~\ref{tab:separation}.}
\label{tab:separation-main}
\end{table}

\section{Greedy vs.\ Sequential Disjoint Greedy: Separation Instances}
\label{app:separation}

We evaluate whether the theoretical containment advantage of sequential disjoint greedy
(Theorem~\ref{thm:half-containment}) translates to practical value differences
on random interference coverage functions
$f(S) = |\bigcup_{i \in S} C_i| - \lambda \sum_{i < j \in S} \mathrm{intf}(i,j)$.
Each element~$i$ covers a random subset $C_i \subseteq [m]$, $|C_i| \in [3,8]$;
pairwise interference arises with probability~0.25, intensity~$\mathrm{Unif}(1,5)$;
$\lambda \sim \mathrm{Unif}(0.5, 2.5)$.
Both algorithms receive the same budget $p = \omega k$ (so $|P_{\textsc{Greedy}}| = |P_{\textsc{SDG}}|$).
After pruning, a $k$-subset is extracted via greedy from each pruned set.

\begin{table}[h]
\centering
\small
\begin{tabular}{@{}cc ccc cc cc@{}}
\toprule
$n$ & $k$ & $\omega$ & Instances & Greedy-$k$ subopt
& $\OPT_k \!\subseteq\! P_g$ & $\OPT_k \!\subseteq\! P_s$
& Val.\ seps & Max gap \\
\midrule
20 & 3 & 2 & 2000 & 39\% & 60\% & 78\% & 123 (6.2\%) & +0.150 \\
30 & 3 & 2 & 1000 & 40\% & 58\% & 73\% & 86 (8.6\%) & +0.136 \\
50 & 3 & 2 & 500 & 34\% & 66\% & 79\% & 42 (8.4\%) & +0.125 \\
50 & 3 & 3 & 500 & 34\% & 67\% & 91\% & 42 (8.4\%) & --- \\
50 & 3 & 5 & 500 & 34\% & 70\% & \textbf{100\%} & 48 (9.6\%) & --- \\
\midrule
30 & 5 & 2 & 500 & 59\% & 41\% & 53\% & 0 & --- \\
50 & 5 & 2 & 200 & 60\% & 36\% & 42\% & 11 (5.5\%) & +0.135 \\
50 & 5 & 3 & 200 & 60\% & 41\% & 67\% & 33 (16.5\%) & --- \\
\bottomrule
\end{tabular}
\caption{Greedy vs.\ SeqDisjoint containment and value separation on random
interference coverage instances.
\emph{Val.\ seps} counts instances where SeqDisjoint's extracted $k$-subset
has strictly higher value.
At $\omega = 2$, containment differences produce measurable value gaps
(up to~$+0.15$).
At $\omega \ge 5$ the gap vanishes: sufficient budget provides enough
alternatives within any pruned set.
Mean $\alpha$ is $\ge 0.96$ for both algorithms across all configurations---the
separation is a tail phenomenon.}
\label{tab:separation}
\end{table}

\section{External Assets, Licenses, and Terms}
\label{app:assets}

We use external assets only as benchmarks or baselines; we do not redistribute
the SNAP, HotpotQA, MuSiQue, COMBHelper, or LLM model assets with this
submission.

\paragraph{SNAP social-network graphs.}
The Facebook and Wiki-Vote graphs are from the Stanford Network Analysis
Project (SNAP) dataset collection~\cite{LeskovecKrevl2014}.  The SNAP
Facebook page identifies the graph as anonymized Facebook social-circle data
from McAuley and Leskovec~\cite{McAuleyLeskovec2012}; the Wiki-Vote page
identifies the graph as Wikipedia adminship vote-history data.  The public
download pages are \url{https://snap.stanford.edu/data/ego-Facebook.html} and
\url{https://snap.stanford.edu/data/wiki-Vote.html}.  We did not find a named
data license on these SNAP dataset pages; accordingly, we use the public
benchmark data only for aggregate MaxCut experiments, cite the original source,
and do not redistribute the raw graph files.

\paragraph{HotpotQA.}
HotpotQA~\cite{YangQiZhang2018} is distributed under the Creative Commons
Attribution--ShareAlike 4.0 International (CC BY-SA 4.0) license; the official
dataset page also states that the processed Wikipedia corpus released with
HotpotQA is under the same license.  We use the distractor setting and
Wikipedia-derived passages as an evaluation benchmark and report only aggregate
metrics.  Official page: \url{https://hotpotqa.github.io/}.

\paragraph{MuSiQue.}
MuSiQue~\cite{TrivediMultiHop2022} is distributed under the Creative Commons
Attribution 4.0 International (CC BY 4.0) license.  The repository also cautions
that MuSiQue composes questions from seed single-hop datasets and releases
dev/test single-hop IDs to avoid train/test leakage.  We use the released
MuSiQue benchmark for evaluation only and do not train on those seed datasets.
Official repository: \url{https://github.com/StonyBrookNLP/musique}.

\paragraph{COMBHelper baseline.}
COMBHelper~\cite{TianEtAl2024} is released under the MIT License in the public
repository \url{https://github.com/1041877801/COMBHelper}.  Our MaxCut
comparison adapts the published architecture for evaluation; if code is
released with a future version, the MIT notice must be retained for any
substantial reused portions.

\paragraph{LLM reader model.}
The reported LLM validation experiments use \texttt{openai/gpt-oss-120b} as a
locally served reader model.  The model card lists the license as Apache 2.0;
the model weights are not redistributed, and only aggregate answer-quality
metrics are reported.  Model card:
\url{https://huggingface.co/openai/gpt-oss-120b}.

\end{document}